%% file: draft.tex
\newif\ifFULL
\FULLtrue
\ifFULL 
\documentclass[prodmode,acmjacm]{full-acmsmall} 
\else

\documentclass[prodmode,license]{acmsmall-ec15} 

\doi{2764468.2764480 }

\fi
\usepackage[ruled]{algorithm2e}
\renewcommand{\algorithmcfname}{ALGORITHM}
\SetAlFnt{\small}
\SetAlCapFnt{\small}
\SetAlCapNameFnt{\small}
\SetAlCapHSkip{0pt}
\IncMargin{-\parindent}

\acmVolume{X}
\acmNumber{X}
\acmArticle{X}
\acmYear{2015}
\acmMonth{2}
\usepackage[numbers]{natbib}
\usepackage[usenames]{xcolor}
\usepackage{xcolor}
\usepackage{amsmath}
\usepackage{threeparttable}
\usepackage[T1]{fontenc}
\usepackage[latin9]{inputenc}
\usepackage{geometry}
\usepackage{verbatim} 
\usepackage{amssymb}
\usepackage{setspace}
\usepackage[justification=centering]{caption}
\usepackage[english]{babel}

\usepackage{hyperref}
\hypersetup{
  colorlinks   = true, 
  urlcolor     = blue, 
  linkcolor    = blue, 
  citecolor   = blue 
}
\usepackage{cleveref}

\allowdisplaybreaks[1] 

  \providecommand{\claimname}{Claim}
  \providecommand{\definitionname}{Definition}
  \providecommand{\lemmaname}{Lemma}
\providecommand{\corollaryname}{Corollary}
\providecommand{\theoremname}{Theorem}

\DeclareMathOperator*{\argmax}{arg\,max}

\providecommand{\tabularnewline}{\\}

\usepackage{complexity}

\DeclareMathOperator{\Ima}{Im}
\newcommand{\e}{\mathrm{e}}

\makeatother

\begin{document}


\newcommand{\f}{\rho}

\newcommand{\hF}{\hat{F}}

\renewcommand{\K}{M}
\renewcommand{\v}{\varepsilon}
\renewcommand{\u}{m}
\renewcommand{\p}{\boldsymbol {\pi}}

\newcommand{\w}[1]{\mathbf{#1}}
\renewcommand{\v}[1]{\mathbf{#1}^*}
\newcommand{\n}[1]{#1_1,\ldots, #1_n}

\newcommand{\chgdel}[1]{\textcolor{red}{\sout{#1}}}
\newcommand{\chgins}[1]{\textcolor{blue}{#1}}

\newcommand{\msout}[1]{\text{\sout{\ensuremath{#1}}}}
\newcommand{\chgdelm}[1]{\textcolor{red}{\msout{#1}}}
\newcommand{\chgdels}[1]{\textcolor{red}{\cancel{#1}}}

\markboth{Chen, Li, Lin, and Rubinstein}{Combining Traditional Marketing and Viral Marketing with Amphibious Influence Maximization}

\title{Combining Traditional Marketing and Viral Marketing with \\ Amphibious Influence Maximization\ifFULL $^*$\fi}
\author{
WEI CHEN
\affil{Microsoft Research, \texttt{weic@microsoft.com}}
FU LI$^1$
\affil{Tsinghua University, \texttt{fuli.theory.research@gmail.com}}
TIAN LIN$^2$
\affil{Tsinghua University, \texttt{lint10@mails.tsinghua.edu.cn}}
AVIAD RUBINSTEIN$^3$
\affil{University of California at Berkeley, \texttt{aviad@eecs.berkeley.edu}}
}

\begin{abstract}
In this paper, we propose the {\em amphibious influence maximization (AIM)} model that combines traditional marketing
	via content providers and viral marketing to consumers in social networks in a single framework.
In AIM, a set of content providers and consumers form a bipartite network while consumers also form their social network, and
	influence propagates from the content providers to consumers and among consumers in the social network following the independent cascade model.
An advertiser needs to select a subset of seed content providers and a subset of seed consumers, such that
	the influence from the seed providers passing through the seed consumers could reach a large number of consumers
	in the social network in expectation.

We prove that the AIM problem is \NP-hard to approximate to within any constant factor via a reduction from
	Feige's $k$-prover proof system for 3-SAT5.
We also give evidence that even when the social network graph is trivial (i.e. has no edges),
	a polynomial time constant factor approximation for AIM is unlikely.
However, when we assume that the weighted bi-adjacency matrix that describes the influence of content providers on consumers is of constant rank,
	a common assumption often used in recommender systems, we provide a polynomial-time algorithm
	that achieves approximation ratio of $(1-1/\e-\varepsilon)^3$ for any (polynomially small) $\varepsilon > 0$.
Our algorithmic results still hold for a more general model where cascades in social network follow a general monotone and submodular function.

\end{abstract}





\begin{bottomstuff}
\ifFULL
$^*$ This is the full version of the paper appeared in ACM EC'2015.
\fi
	
$^1$  This work was mostly done while visitng Microsoft Research Asia. 
This research was supported in part by National Basic Research Program of China Grant 2011CBA00300, 2011CBA00301, and by National Natural Science Foundation of China Grant 61033001, 61361136003.

$^2$ This work was mostly done while interning at Microsoft Research Asia.

$^3$ This work was mostly done while interning at Microsoft Research Asia, and partly at the Simons Institute for the Theory of Computing.
This research was supported by NSF grants CCF0964033 and CCF1408635, and by Templeton Foundation grant 3966. 

%
\end{bottomstuff}

\maketitle

\newcommand{\ignore}[1]{}

\newcommand{\compilehidecomments}{false}

\ifthenelse{ \equal{\compilehidecomments}{true} }{%
	\newcommand{\wei}[1]{}
	\newcommand{\aviad}[1]{}
	\newcommand{\tian}[1]{}
	\newcommand{\fu}[1]{}
}{
	\newcommand{\wei}[1]{{\color{blue}  [\text{Wei:} #1]}}
	\newcommand{\aviad}[1]{{\color{brown} [\text{Aviad:} #1]}}
	\newcommand{\tian}[1]{{\color{gray}  [\text{Tian:} #1]}}
	\newcommand{\fu}[1]{{\color{purple} [\text{Fu:} #1]}}
}

  \newtheorem{defn}{\protect\definitionname}
\newtheorem{thm}{\protect\theoremname}
  \newtheorem{lem}{\protect\lemmaname}
  \newtheorem{cor}{\protect\corollaryname}
  \newtheorem{claim}{\protect\claimname}

\renewcommand{\H}{H}
\newcommand{\Mj}{M_{\cdot j}}
\renewcommand{\S}{\ensuremath{{\cal S}_\varepsilon}}
\renewcommand{\R}{R}
\renewcommand{\b}{\lambda}

\newcommand{\U}{\rho}
\renewcommand{\E}[1]{\mathrm{E}\left[#1\right]}

\input{intro}
\input{prelims}

\input{2layers}
\input{3layers}
\input{constant_rank}



%
\bibliographystyle{acmsmall}
\bibliography{draft}



\ifFULL

\appendix


\input{appendix-definition}

\input{appendix-boost}

\input{appendix-submodularity}
\input{appendix-APM}

\fi

\end{document}

%% file: intro.tex

\section{Introduction}

Marketing is traditionally partitioned into several stages: advertisers pay content providers (e.g. TV networks, radio stations, online news sites,  influential bloggers, etc.); content providers recruit audience; and then the audience who are exposed to the advertisements influence their friends. 
Today, with the development of the Internet and social networks, 
	there is an enormous amount of data that can be used to
	predict which users will enjoy a specific content,
	which users are likely to purchase the advertised product, and which users can influence their friends to buy the product as well. 
More importantly, information is available to track the individuals 
	who participate in each one of those interactions. 
This suggests a new marketing approach in which advertisers can contact both
	content providers and the audience at the same time, with the goal of
	maximizing the overall exposure 
	(through direct exposure as well as propagation via social networks) 
	to the advertisement.

Consider the following example.
Suppose a technology company wants to select a subset of regular tech bloggers 
	(content providers) 
	and engage them with marketing activities so that they would cover
	the company extensively and favorably. 
However, this alone does not guarantee that these favorable blogs can reach
	the targeted customers of the company.
The company may further select a number of non-bloggers and spend its marketing
	effort on them (e.g. buying advertising slots to remind them about
	the blog entries of their selected bloggers) to make them 
	active in subscribing, reading, and propagating the blog entries 
	written by the company's selected bloggers.
The objective of the company is to maximize
	the number of targeted customers who get exposed to the favorable blogs, either directly
	or indirectly through links forwarded by friends in the social network.

The above proposed marketing strategy can be viewed as a combination of traditional marketing via
	content providers and viral marketing in social networks.
It can be modeled as a controlled diffusion
	in a joint network consisting of a bipartite graph modeling provider-consumer
	relationship and a social graph modeling social influence relationship among
	the consumers.
The bipartite graph and its edge weights indicate the influence from content
	providers to consumers, while the social graph and its edge weights indicate
	the influence among consumers.
An advertiser wants to select a subset of content providers (called {\em seed providers})
	and a subset of consumers (called {\em seed consumers})
	in the social network such that the influence from
	seed providers could activate enough seed consumers, which in turn
	could activate more consumers in the social network.
Since the two marketing activities involve costs of different types, 
	we enforce separate budgets on provider selection and consumer selection.

In this paper, we model the above combined marketing strategy as the following
	{\em amphibious influence maximization (AIM)} problem.
We are given (a) a bipartite graph $B=(U,V,M)$ where $U$ represents content providers,
	$V$ represents consumers, and $M$ is the weighted bi-adjacency matrix representing
	the influence probabilities from providers to consumers;
	and (b) a directed social graph $G=(V,P)$ where $V$ is the same set of consumers as in $B$ and
	$P$ is the weighted adjacency matrix representing influence probabilities of each consumer over her friends.
Given a subset $X\subseteq U$ of seed providers and a subset $Y \subseteq V$ of seed consumers,
	the influence propagates from $X$ to $Y$ and then to other consumer nodes in $V$ following
	the independent cascade model \cite{kempe2003maximizing}.
Given budgets $b_1$ for providers and $b_2$ for consumers, the AIM problem is to
	select at most $b_1$ seed providers and $b_2$ seed consumers such that 
	the expected number of activated consumer nodes after the diffusion process is maximized.

One important nature of the AIM problem formulation is that seed consumer selection is non-adaptive.
That is, we need to select seed providers and seed consumers 
	together before we observe the actual cascades from the seed providers.
This is motivated by long-term marketing campaigns, during which repeated cascades may be
	generated from content providers.
For such campaigns it is impractical for advertisers to adaptively select
	seed consumers for every cascade, and thus non-adaptive seed consumer selection
	aiming at maximizing the cumulative effect over multiple cascades is desirable.
\ignore{
This is suitable for the marketing campaigns that are meant to be executed over a period of time
	during which influence cascades may be generated repeatedly from the seed providers.
In such campaigns the advertisers cannot afford to dynamically change seed consumers for every
	cascade from the seed providers, and thus it is appropriate to choose non-adaptive solutions.
}

\subsection*{Our results}

We study both the hardness of the AIM problem and its approximation algorithms
	in the independent cascade (IC) model \cite{kempe2003maximizing}.
In terms of hardness, we warm up (Section \ref{sec:Two-layers}) with an easy result
	that finding any constant-factor approximation for AIM
	(even when the social network graph has no edges at all)
	is as hard as approximating the densest-$k$-subgraph problem,
	for which no polynomial-time algorithm is known.
Our main impossibility result (Section \ref{sec:Three-layers}) is that AIM is 
	also \NP-hard to approximate to within any constant factor.
The result is proven by a reduction from Feige's $k$-prover proof 
	system for 3-SAT5 \cite{feige98}. 

In order to overcome the above strong inapproximability results,
	we introduce additional assumptions in our model.
Both hardness reductions construct a providers-consumers bipartite graph $B$
	with a complex and elaborate structure. 
In practice, even if the true relationship is indeed so intricate in nature,
	most of the learning techniques that are used to estimate this relationship
	assume some simple underlying structure - 
	so we can expect the input for our algorithm to be "simple".
In particular, for the specific motivation of influence of content providers on consumers,
	a common assumption in the construction of the influence matrix is that it is (approximately) low-rank (e.g. the "Netflix Problem"; \cite{KBV09-netflix}).
This assumption is typically motivated by modeling the relationship between content and consumers
	via a small number of (hidden) features.
In Section \ref{sec:Constant-rank-adjacency} we show that 
	when the weighted bi-adjacency matrix $P$ has constant rank,
	we can approximate AIM to within a factor of $(1-1/\e-\varepsilon)^3$ in polynomial time for any (polynomially small) $\varepsilon > 0$.
Our algorithmic result can be generalized to accommodate any diffusion model in the social network
	that has a monotone, submodular, and polynomial-time computable influence spread function.
\ignore{
In order to overcome the above hardness results, we assume that the weighted bi-adjacency matrix $M$ of the 
	bipartite graph has constant rank $r$,
	which can be justified by some empirical studies \cite{xxx}.
\wei{Aviad may need to add some citations to justify the constant rank
	assumption.}
We show that for constant rank $r$ matrix $M$, for any $\varepsilon>0$, we can approximate
	AIM to within $(1-1/\e)^2(1-1/\e-\varepsilon)$ in polynomial time.
\wei{We may need to replace $(1-1/\e)$ with $(1-1/\e-\varepsilon)$ and add with high probability.}
}


\subsection{Related work}

Influence maximization is first studied as an algorithmic problem with application to viral marketing
	by \citet{domingos2001mining,richardson2002mining}.
\citet{kempe2003maximizing} first formulate it as a discrete optimization problem.
They summarize the independent cascade model and linear threshold model, and apply submodular function maximization to
	obtain approximation algorithms for influence maximization.
Extensive research has been done since to improve the scalability of the algorithm, extending the model to
	competitive setting, etc. (cf. \cite{CLC13}).

Conceptually, amphibious influence maximization combines viral marketing with traditional marketing via content providers, and thus
	it enriches viral marketing and its technical formulation of influence maximization to a new level.
Technically, AIM also contains influence maximization as a special case: when 
	we have provider budget $b_1 = |U|$ (allowing all providers to be seeds) and bi-adjacency matrix to be all-one matrix
	(providers would deterministically activate all seed consumers), AIM is reduced to the classical influence maximization problem.

Recently, Seeman and Singer initiated a line of works \cite{SeemanS13, BPRSS, RSS14-knapsack-seeding}
on adaptive seeding in social networks that is closely related to ours.
In the adaptive seeding problem, a small subset $X$ of the nodes in a social network is initially available to an advertiser.
In the first stage, the advertiser selects (or {\em seeds}) a subset $S$ of these nodes, who may influence some of the neighbors.
In the second stage, a random subset of the neighbors of $S$ becomes available; 
the advertiser spends the rest of her budget on seeding a subset of the newly available nodes,
in hope to maximize their influence in the social network. 
The most important difference between Seeman and Singer's model and ours is that in the former,
the seeding is {\em adaptive}, i.e. the advertiser waits to see which of the second layer's nodes became available before selecting a subset.
Recall in our model, per contra, the advertiser must seed consumers in advance;
in particular, there is no guarantee that after the edge percolation, 
a seed consumer $y \in Y$ will have live edges with seed content providers.
As already discussed, non-adaptive seeding is appropriate for marketing campaigns during which repeated influence cascades may occur. 

From a technical viewpoint, although we certainly build on ideas from \cite{BPRSS, RSS14-knapsack-seeding},
adaptivity completely changes the approximability of the problem:
all the works above achieve constant-factor approximations in different settings of adaptive seeding,
while we show that in the non-adaptive case, constant-factor approximation is impossible%
\footnote{Note that this is a comparison of the algorithmic limitations within each model, and not a competitive analysis.
In particular, whenever seeding adaptively is feasible, it is of course preferable and can perform much better than "non-adaptive seeding".
As mentioned earlier, our motivation for studying a non-adaptive model is 
settings where the time required to estimate long-term influence of a marketing campaign makes adaptive seeding impractical.}.
Interestingly, all the above works on adaptive seeding use a non-adaptive relaxation of the adaptive problem.
It turns out that unlike the non-adaptive AIM problem, 
{\em the non-adaptive relaxation can be approximated efficiently to within a constant factor}.
(The precise factor of approximation depends on other parameters of the problem such as IC model vs. a general submodular function.)

The problem of acceptance probability maximization (APM) for active friending studied by \citet{YHLC13} is also related to our work.
In APM, a source node needs to select $k$ intermediary nodes situated between the source and the target node 
	in a social network, such that
	if influence from the source only propagates through intermediaries, the probability of activating the target is maximized.
AIM and APM are similar in that both need to select some intermediary nodes between the source and the target and both study the
	non-adaptive version.
However, their assumptions
	on influence cascade are different: APM assumes that cascades only occur in the sub-network 
	consisting of the source, the selected intermediaries and the target,
	while AIM assumes that cascades occur from the selected sources to the selected intermediaries but from intermediaries cascades can reach
	the entire social network.
For APM problem, \citet{YHLC13} only provide a heuristic algorithm and do not have hardness of approximation results. 
\ifFULL
In Appendix \ref{sec:APM},
\else
In our full report \cite{CLLR15}, 
\fi
we build on the hardness of approximation of AIM to prove
	that it is \NP-hard to approximate APM in a general graph
	to within a near-exponential ($2^{n^{1-\varepsilon}}$) factor.

Finally, our algorithm for AIM with constant rank was inspired by recent works
	that (approximately) solve the Densest-$k$-Bi-Subgraph problem in graphs
	with (approximately) constant rank \cite{ALSV13-approx_rank, PMDC14-constant_rank}.

%

%% file: prelims.tex
\section{Model and Problem Definition}
\label{sec:model}

We consider a (heterogeneous) network consisting of the following two components.
The first is a bipartite graph $B = (U,V,\K)$, where $U$ represents content providers (e.g.
	bloggers, TV programs, etc.), $V$ represents consumers, and $\K$ is the $|U|\times |V|$
	weighted bi-adjacency matrix with $\K_{ij} \in [0,1]$ denoting the probability that
	$i\in U$ would successfully activate $j\in V$ (e.g. $j$ is influenced by the 
	advertisement associated with $i$).
The second is a directed social graph $G=(V,P)$, where $V$ is the same as the $V$ in the bipartite
	graph $B$, 
	and	$P$ is the $|V|\times |V|$ weighted adjacency matrix 
	with $P_{vw}$ denoting the influence probability
	from $v\in V$ to $w\in V$.
We denote the set of directed edges of the social graph as $E=\{(v,w)\mid P_{vw}>0 \}$.

After fixing a set of seed providers $X\subseteq U$ and a set of seed consumers $Y\subseteq V$,
	we model the influence diffusion from $X$ to the nodes in the social graph $G$ as follows.
For each edge $(i,j)$ in $B$ we sample it as {\em live} with probability $\K_{ij}$ and
	{\em blocked} with probability $1-\K_{ij}$; for each edge $(v,w)\in E$, we sample
	it as live with probability $P_{vw}$ and blocked with probability $1-P_{vw}$.
We say that a node $v\in V$ is {\em activated (by the influence of $X$ through $Y$)} if
	there is a path $(x,y,v_1,\ldots, v_t=v)$ with $x\in X$ and $y\in Y$, and all edges
	on the path are live.
Given $X$ and $Y$, we use $\sigma(X,Y)$ to denote the expected 
	number of activated nodes in $V$ (with expectation taken among all samples on 
	all edges), and call it the {\em influence spread} of $X$ and $Y$.

Note that the diffusion model can be equivalently described as follows.\footnote{Equivalence
	is in the sense of the distribution of final set of activated nodes in V.}
First, every seed $i \in X$ independently tries to activate every node 
	$j \in Y$ with success probability $\K_{ij}$, and $j \in Y$ is activated as long as some
	$i \in X$ activates $y$, and nodes outside $Y$ are not activated by seeds in $X$.
Let  $S\subseteq Y$ be the (random) set of nodes activated in $Y$.
Then we treat $S$ as the seed set and apply the independent 
	cascade model~\cite{kempe2003maximizing} to start the influence diffusion from 
	$S$ in the social network $G$ using influence probabilities $P$: namely at each discrete
	time step, each newly activated node $v\in V$ has one chance to activate each of its
	outgoing neighbor $w\in V$ with probability $P_{vw}$.

Our goal is to find a set $X$ of seed providers of size $b_1$ and a set
	$Y$ of seed consumers of size $b_2$ such that they work together to generate the largest influence
	spread, which we formally define below.

\begin{defn} [Amphibious Influence Maximization] \label{def:AIM}
In the {\sc Amphibious Influence Maximization} ({\sc AIM}) problem, 
we are given a bipartite graph $B = (U,V,\K)$ and a directed social graph 
	$G=(V,P)$, and budgets $b_1$ and $b_2$, and we want to find 
	a subset $X^*\subseteq U$ of size $b_1$ and a subset $Y^*\subseteq V$ of size $b_2$
	such that the influence spread of $X^*$ and $Y^*$ are maximized, that is,
	finding $X^*$ and $Y^*$ such that
\[
(X^*,Y^*) = \argmax_{X\subseteq U, |X|=b_1, Y\subseteq V, |Y|=b_2} \sigma(X,Y).
\]
\end{defn}

Several remarks are now in order.
First, when we set $b_1=|U|$ and $\K$ as an all-one matrix, the AIM problem is reduced to 
	the classical influence maximization problem defined in~\cite{kempe2003maximizing}.
Thus, AIM is a generalization of the classical influence maximization problem such that
	it considers interactions between the provider nodes $U$ and consumer nodes $V$ and they
	have to work together to spread the influence.
Second, it is easy to see that
	when either fixing set $X$ or $Y$, $\sigma(X,Y)$ as a set function of the other variable
	is monotone and submodular.\footnote{A set function $f$ is monotone if
		for all $S\subseteq T$, $f(S)\le f(T)$, and submodular if 
		for all $S\subseteq T$ and $v\not\in T$, $f(S\cup \{v\}) - f(S) \ge 
		f(T\cup \{v\})- f(T)$.}
However, the interaction of $X$ and $Y$ makes the AIM problem much harder
	than the classical influence maximization problem:
	we need both nodes in $X$ and $Y$ to generate 
	influence and missing either of them will not work.
Finally, our results can be generalized to allow 
	diffusion models in the social network to follow any monotone and submodular function, and non-seed consumers to be influenced
	with background probabilities.
To simplify the presentation, we focus on the main problem given in Definition~\ref{def:AIM}
	and discuss the generalization in Section~\ref{sec:general}.

%% file: 2layers.tex

\section{Hidden-clique hardness}\label{sec:Two-layers}

Before we derive our main hardness result,
we briefly describe in this section a much simpler reduction which gives a weaker hardness, "Hidden-clique hardness" (sometimes also "planted-clique").
Another feature of this result is that in the hard instance
the social network graph $G$ has no edges at all!

\paragraph{Hidden clique}

In an Erdos-Renyi random graph ${\cal G}\left(n,1/2\right)$ the largest
clique size is approximately $2\log_{2}n$, with high probability (e.g. \cite{AS92-probabilistic_method}).
We can ``plant'' a clique of size $t\gg2\log n$,
by choosing $t$ nodes at random, and connecting all the edges between
them. The hidden clique problem (e.g. \cite{densest_k-subgraph_AAMMW11})
is to distinguish between a graph sampled from ${\cal G}\left(n,1/2\right)$
and a graph from ${\cal G}\left(n,1/2\right)$ with a planted clique.
Alon et al. \cite{densest_k-subgraph_AAMMW11} reduce this problem
to solving the following gap version of {\sc Densest $k$-Subgraph}.
Although the planted clique problem
has been extensively studied, the best known
algorithms run in quasi-polynomial time ($n^{O\left(\log n\right)}$);
in particular, there are no known polynomial-time algorithms for the
hidden clique problem.
\begin{theorem}
(Theorem 1.3 of \cite{densest_k-subgraph_AAMMW11}) If there is no
polynomial-time algorithm for the hidden clique problem with a planted
clique of size $t=n^{1/3}$, then for any $\delta>0$, there is no
polynomial-time algorithm that given a graph $G$ distinguishes between:
\begin{description}
\item [{Completeness}] $G$ has a clique of size $k$; and
\item [{Soundness}] Every $k$-subgraph of $G$ has density at most $\delta$.
\end{description}
\end{theorem}

Hardness for {\sc AIM} follows as a corollary:

\begin{corollary}
If there is no polynomial-time algorithm for the hidden clique problem
with a planted clique of size $t=n^{1/3}$, 
then {\sc AIM} cannot be approximated to within a constant factor in polynomial time - even in the special case where the social network graph has no edges.
\end{corollary}

\begin{proof}
We give a reduction from the {\sc Densest $k$-Subgraph} problem.
\begin{description}
\item [{Reduction}] Given an instance $G^{DkS}=\left(V^{DkS},E^{DkS}\right)$  of {\sc Densest $k$-Subgraph}
	with gap parameter $\delta$,
we construct an {\sc AIM} weighted bipartite graph $B=\left(U,V,M\right)$ 
between content provider nodes and consumer nodes as follows: 
We identify both $U$ and $V$ with the original set of vertices $V^{DkS}$
(i.e. our {\sc AIM} instance has twice as many vertices).
For any $u\in U$ and $v \in V$, we set $M_{u,v} = 1/n^2$ 
if the corresponding vertices in $G^{DkS}$ are distinct and have an edge between them,
and $M_{u,v} = 0$ otherwise. 
We set the budgets to $b_1 = b_2 = k/2$.

\item [{Completeness}] If $G^{DkS}$ contains a $k$-clique $C\subseteq V$,
then partition $C$ into two subsets of size%
\footnote{We assume without loss of generality that $k$ is even. Given a polynomial
time algorithm for even $k$ it is easy to extend to an algorithm
for $k-1$; e.g. by adding a dummy vertex that is connected to all
vertices in the graph.}
$k/2$ and label them $X'$ and $Y'$. 
Consider their respective copies $X\subseteq U$ and $Y\subseteq V$
in the bipartite graph: it follows from the construction that $X\cup Y$
is a bi-clique.
Thus, every consumer in $Y$ has probability $1-\left(1-1/n^2\right)^{k/2} = \left(1-o(1)\right) \cdot k/\left(2n^2\right)$.
Summing over all $k/2$ consumers, the expected number of activated nodes is $OPT = \left(1-o(1)\right) \cdot k^2/\left(4n^2\right)$.
\item [{Soundness}] 
Let $X,Y$ be an optimal solution of the {\sc AIM} instance.
Let $S \subseteq V^{DkS}$ be the union of the copies of $X$ and $Y$ in $G^{DkS}$.
By the premise, $S$ contains at most $\delta {{k}\choose{2}}$ edges.
Thus there are at most $2\delta {{k}\choose{2}}$ edges between $X$ and $Y$,
each with weight $1/n^2$ 
(we may count some edges twice in case the copies of their endpoints belong to both $X$ and $Y$).
Therefore, $\sigma(X,Y) < \delta k^2 / n^2 < 5\delta\cdot OPT$.
\end{description}
\end{proof}

\ignore{

\begin{proof}
We reduce {\sc Densest $k$-Subgraph} to {\sc Densest Bipartite
$k$-Subgraph}, the restricted problem for bipartite graphs. The
reduction to the special case of {\sc AIM} follows by taking all the edges in the
bipartite graph $B$
to have sufficiently small probability $p=o\left(1/n\right)$ and removing
	all edges in the social graph, since in this case
	the probability that a node in the social graph is activated is approximately
	its degree in the bipartite graph multiplied by $p$.
\begin{description}
\item [{Reduction}] Given a graph $G=\left(V,E\right)$ we construct $G^{B}=\left(V^{B},E^{B}\right)$
as follows: let $V^{1},V^{2}$ be two copies of $V$, and set $V^{B}=V^{1}\cup V^{2}$.
For any $u^{1}\in V^{1}$ and $v^{2}\in V^{2}$, let $u$ and $v$,
respectively be the corresponding vertices in $V$. We define $E^{B}$
to have an edge $\left(u^{1},v^{2}\right)$ if and only if the corresponding
edge $\left(u,v\right)$ is in $E$. Observe that since $\left(u,v\right)$
also induces the edge $\left(u^{2},v^{1}\right)$, the number of edges
in the bipartite graph also doubles: $\left|E^{B}\right|=2\left|E\right|$.
\item [{Completeness}] If $G$ contains a $k$-clique $C\subseteq V$,
then partition $C$ into two subsets of size%
\footnote{We assume without loss of generality that $k$ is even. Given a polynomial
time algorithm for even $k$ it is easy to extend to an algorithm
for $k-1$; e.g. by adding a dummy vertex that is connected to all
vertices in the graph.%
} $k/2$ and label them $C_{1}$ and $C_{2}$. Consider their respective
copies $C_{1}^{1}\subseteq V^{1}$ and $C_{2}^{2}\subseteq V^{2}$
in the bipartite graph: it follows from the construction that $C_{1}^{1}\cup C_{2}^{2}$
is a bi-clique.
\wei{Just checking: for bipartite graph, the above is still called $k$-subgraph, not
$k/2$-subgraph. I guess so.}
\aviad{I think that it is a $k$-subgraph and a $k/2$-bi-subgraph. Adapting the reduction to the other definition is easy (just add an edge between each vertex and its copy).}
\item [{Soundness}] Suppose that $G^B$ 
has a subgraph $\left|S^{1}\cup T^{2}\right|=k$
of density greater than $2\delta$, i.e. $\left|E^{B}\cap\left(S^{1}\times T^{2}\right)\right|>2\delta{k \choose 2}$.
\wei{Just checking: even for bipartite graph, the density of a $k$-subgraph is still
with respect to ${k \choose 2}$.}
\aviad{Again, this is a matter of whether we're looking for the "densest subgraph in a bipartite graph" or the "densest bi-subgraph".
I don't feel strongly about it either way.}
Consider their respective copies $S,T\subseteq V$ ($S$ and $T$
are not necessarily disjoint). Each edge in $E\cap\big(\left(S\cup T\right)\times\left(S\cup T\right)\big)$
corresponds to at most two edges in $E^{B}\cap\left(S^{1}\times T^{2}\right)$.
Therefore $\left|E\cap\big(\left(S\cup T\right)\times\left(S\cup T\right)\big)\right|\geq\frac{1}{2}\left|E^{B}\cap\left(S^{1}\times T^{2}\right)\right|>\delta{k \choose 2}$.
Let $W$ be any subset of $V$ that includes $S\cup T$ and has size
$k$. Then the density of $W$ is also greater than $\delta$: $\left|E\cap\left(W\times W\right)\right|\geq\left|E\cap\big(\left(S\cup T\right)\times\left(S\cup T\right)\big)\right|>\delta{\left|W\right| \choose 2}$.
\end{description}
\end{proof}
}

%% file: 3layers.tex

\section{\NP-hardness of approximation}
\label{sec:Three-layers}

\ignore{
{\sc Three-Layer Non-Adaptive Max-Cover}: Given a tripartite graph
over $\left(U,V_{1},V_{2}\right)$ with probabilities $p_{e}\in\left[0,1\right]$
on the edges, find subsets $T_{i}\subseteq U_{i}$ of total size $B$
that maximize the expected number of nodes in $T_{3}$ that have a
path from $T_{1}$ via $T_{2}$:
\begin{flalign*}
\max & f\left(T_{1},T_{2},T_{3}\right)=\sum_{c\in T_{3}}1-\prod_{b\in T_{2}}\left(1-p_{b,c}\cdot\left(1-\prod_{a\in T_{1}}\left(1-p_{a,b}\right)\right)\right)\\
 & \sum\left|T_{i}\right|\leq B
\end{flalign*}
}

In this section we prove our main hardness result, namely:

\begin{theorem} \label{thm:np-hardness}
 {\sc AIM} is \NP-hard to approximate to within any constant factor. 
\end{theorem}

\paragraph{Proof outline}
We reduce from Feige's $k$-prover proof system \cite{feige98}. The provers' answers
to questions correspond to the provider nodes in $U$.
The provider nodes are connected to a subset $V_1 \subset V$ of the consumers,
on which the verifier can test the provers' answers.
Since the edges from $U$ to $V_1$ appear with low probability,
it is significantly more cost-effective to select a few nodes from
$V_1$ 
with many neighbors in $U$. Intuitively, this corresponds
to a verifier's test which many provers would pass. 
By Theorem \ref{thm:feige},
if we start from a satisfiable formula, all $k$ provers will agree
- versus less than $2$ provers that agree for an unsatisfiable formula.
The rest of the consumers, $V_2 = V \setminus V_1$, 
have incoming edges from influential consumers in $V_1$. 
They will guarantee that the provers answer (almost) all the verifier's questions.

Notice that our hard instance is a three-layered graph. 
We henceforth call the provider nodes the {\em top layer},
the influential consumers $V_1$ constitute the {\em middle layer},
whereas the {\em bottom layer} has the rest of the nodes.

\paragraph{$k$-prover proof system}
Consider $k$ provers trying to prove the satisfiability of some 3-SAT5
formula over $n$ variables.
A 3-SAT5 formula is a conjunctive-normal-form (CNF) formula where each variable appears in exactly $5$ clauses,
and each clause contains exactly $3$ variables; 
notice that there are $5n/3$ clauses.
The verifier selects $l$ clauses (with replacement) uniformly and independently at random.
For each clause, the verifier selects one of the participating variables uniformly and independently at random;
we call those the {\em distinguished} variables.
Each question consists of $l/2$ clauses, and $l/2$ distinguished variables 
	from the remaining $l/2$ clauses.
An answer $a$ to question $q$ consists of assignments to the $l/2+3l/2=2l$ variables
in question. 
Let $R$ be the set of random strings, and $Q$ be the set of
	questions.
For each random string $r\in R$, we associate a question
$q\in Q$ for each prover $i\in\left[k\right]$; we denote this as $(q,i) \in r$.
We henceforth abuse notation and also use $R$ and $Q$ to
denote the corresponding cardinalities $R=n^{l}\cdot5^{l}$ and $Q=n^{l}\cdot\left(\frac{5}{3}\right)^{l/2}$.
Given the $k$ provers' answers, Feige's verifier tests the provers
answers by comparing their answers on the $l$ distinguished variables.
(We will diverge from Feige's construction at this point and use a
stronger test that compares the provers' answers on all $3l$ variables.)
For constant $l$, Feige proves the following theorem.
\begin{theorem}
\label{thm:feige}(Lemma 2.3.1 in \cite{feige98}) Given
a $k$-prover system, it is \NP-hard to distinguish a 3-SAT5 formula between the following:
\begin{description}
\item [{Completeness}] all the provers pass the verifier's test with probability
$1$; and
\item [{Soundness}] the probability that any pair of provers pass the verifier's
test is at most $2^{-cl}$, for some constant $c>0$.
\end{description}
\end{theorem}

\paragraph{Construction}

We construct a directed graph with three layers: $U,V_1,V_2$.
The first layer $U$ is precisely the set of "content providers" in our model.
The set of "consumers" populates the middle and bottom layers $V = V_1 \cup V_2$.
In terms of the weighted bi-adjacency matrix, 
the layered structure means that $M_{u,v} = 0$ for all $u\in U$ and $v \in V_2$,
and similarly $P_{v_1,v_2} = 0$ unless $v_1 \in V_1$ and $v_2 \in V_2$. 

Going back to the $k$-prover system, 
the top layer $U$ corresponds to triplets of provers' answers
to questions; the middle layer $V_{1}$ corresponds to assignments
to variables -{\em distinguished and non-distinguished}- that may
appear in the verifier's question to any of the provers; finally,
the bottom layer corresponds to the random strings of the verifier.
All the edges go from the top to the middle layer, or from the middle to the bottom layer. 
In particular, the graph is tri-partite.

More specifically, for each triplet $\left(q,a,i\right)$ of (question,
answer, prover) we have a corresponding node in $U$. 
For each pair $\left(r,\overline{a_{r}}\right)$
of (verifier's random string, assignment to all $3l$ variables) we
have a node in $V_{1}$. 
Notice that this is different from \cite{feige98},
where the elements to be covered correspond to $\left(r,a_{r},i\right)$
with $a_{r}$ being the assignment only for the distinguished variables.
The $\left(q,a,i\right)$ node is connected
to all the nodes $\left(r,\overline{a_{r}}\right)$ such that: $\left(q,i\right)\in r$,
and when restricting $\overline{a_{r}}$ to the variables specified
by $\left(q,i\right)$, it is equal to $a$. In particular, for each
$i$, each $\left(r,\overline{a_{r}}\right)$ corresponds to only
one $\left(q,a,i\right)$. 
We set the top-layer budget to be the number of nodes in $U$ that correspond to a single assignment, $b_1 = kQ$; 
similarly we let $b_2 = R$ represent the number of nodes in $V_1$ that match the same assignment.
Finally, all the edges from $U$ to $V_{1}$ 
have probability $1/k$.

For each random string $r$, 
we have $\eta$ nodes in the bottom layer, $V_{2}$.
We choose a sufficiently large $\eta$ 
to ensure that most of the utility comes from the bottom layer.
The nodes corresponding to each $r$ are connected to all the nodes 
$\left(r,\overline{a_{r}}\right)$ in $V_{1}$ with probability $1$. 
The role of this layer is to force any good assignment to spread 
its budget across the different random strings
(i.e. make sure that the provers answer all the questions).

See Table \ref{tab:Summary-of-notation} for a summary of notation.

\begin{table}
\caption{Summary  of notation in main reduction}\label{tab:Summary-of-notation}
\vspace{0.7cm}
\begin{tabular}{|c|c|c|}
\hline
Notation & Interpretation in $k$-provers system & Vertices in {\sc AIM}\tabularnewline
\hline
\hline
$\left(q,a,i\right)$ & question, answer, prover & $1$ vertex in $U$\tabularnewline
\hline
$\left(q,i\right)$ & question, prover & $2^{3l/2}$ vertices in $U$\tabularnewline
\hline
$r$ & random string & %
\begin{tabular}{c}
$k\cdot2^{3l/2}$ vertices in $U$\tabularnewline
($\forall$ $\left(q,i\right)\in r$ and $a\in\left\{ 0,1\right\} ^{3l/2}$)\tabularnewline
\end{tabular}\tabularnewline
\hline
$\left(r,\overline{a_{r}}\right)$ & random string, assignment to all $3l$ variables & $1$ vertex in $V_{1}$\tabularnewline
\hline
$r$ & random string & $\eta$ vertices in $V_{2}$\tabularnewline
\hline
$\left(r,h\right)$ & random string, copy & $1$ vertex in $V_{2}$\tabularnewline
\hline
\end{tabular}
\end{table}

\paragraph{Completeness}

Given a satisfiable assignment to the 3SAT-5 formula, 
we select in the top layer a subset $S\subset U$
of $kQ$ nodes that correspond to the same assignment.
Because they all correspond to the same assignment, 
for each random string $r$, 
all $k$ corresponding nodes in $S$ are connected
to the common node $\left(r,\overline{a_{r}}^{*}\right)$. 
In the middle layer, we let $T$ be the set of these $R$ nodes 
(i.e. $\left(r,\overline{a_{r}}^{*}\right)$ for $r \in R$).
Before sampling the edges, each $\left(r,\overline{a_{r}}^{*}\right)$
has $k$ neighbors in $S$. After sampling, the probability
that there is a path from $S$ to $\left(r,\overline{a_{r}}^{*}\right)$
is $1-\left(1-\frac{1}{k}\right)^{k}\approx1-1/e$.

Since each node in $T$ has $\eta$ neighbors in $V_2$ (with probability $1$),
the value of this solution is approximately $OPT \approx (1-1/e) R\eta$.

\paragraph{Soundness}

In an unsatisfiable instance, any two provers agree for at most a
$\left(2^{-cl}\right)$-fraction of the random strings. 
We will show
in Lemma \ref{lem:good-random-strings} that there are 
at most $\left(2\cdot2^{-\left(1/3\right)cl} \cdot R\right)$ 
{\em good} random strings $r$, which are strings $r$ such that there is a node $\left(r,\overline{a_{r}}\right)$
with more than one neighbor in $S$. 
Since for each random
string $r$ there are only $\eta$ nodes in $V_{2}$, each of the good random
strings contributes at most $\eta$ to the value of the solution. 
Before sampling the edges, 
any node that does not correspond to a good random string
has at most one neighbor in $S$. 
After sampling, the probability that any such node has a neighbor in $S$ is at most $1/k$. 
Since each node in $V_{1}$ has $\eta$ neighbors in $V_{2}$, 
the total contribution from $R$ nodes that do not correspond to good
random strings is bounded by $R\eta /k$. 
Therefore, the expected number of covered nodes is bounded by 
the contribution of the middle layer, plus the contributions from the good and bad random strings:
\[
R + \left(2\cdot2^{-\left(1/3\right)cl}\cdot R\right) \eta +R\eta /k= \left(1/k+o\left(1\right)\right)R\eta \approx \left(\frac{e}{e-1}\cdot \frac{1}{k}\right)OPT\mbox{.}
\]
Lemma~\ref{lem:good-random-strings} below completes of Theorem~\ref{thm:np-hardness}.

\begin{lemma}
\label{lem:good-random-strings}There are at most $\left(2\cdot2^{-\left(1/3\right)cl} \cdot R\right)$
good random strings.\end{lemma}
\begin{proof}
Intuitively, any $\left(r,\overline{a_{r}}\right)$ which has more than one
neighbor in $S$ corresponds to an agreement of at least two
provers - and therefore should be a rare event. 
In order to turn this intuition into a proof, 
we must rule out solutions that distribute
the budget in an uneven manner that does not correspond to answers
of honest provers to verifier's questions. 

In expectation,
for each $\left(q,i\right)$ there is only one $\left(q,a,i\right) \in S$.
Therefore by Markov's inequality, 
for at most a $2^{-\left(1/3\right)cl}$-fraction of $\left(q,i\right)$'s,
more than $2^{\left(1/3\right)cl}$ corresponding nodes belong to $S$;
we call those $\left(q,i\right)$'s {\em heavy}, and {\em light} otherwise, i.e., 
\begin{gather}
\Pr_{r}\left[\exists i:\mbox{ \ensuremath{\left(q,i\right)}\,\ is heavy}\right]\leq2^{-\left(1/3\right)cl}\cdot k\mbox{.} \label{eq:markov1}
\end{gather}
We henceforth focus on bounding the number of good random
strings that correspond only to light $\left(q,i\right)$'s.

Consider only $r$'s whose $\left(q,i\right)$'s are light.
For each light $\left(q,i\right)$, 
there are at most $2^{\left(1/3\right)cl}$ nodes $\left(q,a,i\right)$ in $S$. 
In other words, each prover submits at most $2^{\left(1/3\right)cl}$ answers
to each question. 
By Theorem \ref{thm:feige}, if each prover submits only one answer to each question,
the fraction of random strings for which at least one pair agrees is at most $2^{-cl}$;
having $2^{\left(1/3\right)cl}$ answers, the probability that any pair agrees increases by at most $2^{\left(2/3\right)cl}$.
Therefore at most a $2^{-\left(1/3\right)cl}$-fraction
of random strings have at least one pair of agreeing answers. 
Recall that a random string $r$ is {\em good} if for some $\overline{a_{r}}$,
the node $\left(r,\overline{a_{r}}\right)$ has more than one neighbor
$\left(q,a,i\right)$ in $S$.
\begin{gather}
\Pr_{r}\left[\mbox{\ensuremath{r}\ is good and \ensuremath{\forall i: \left(q,i\right)} is light}\right]\leq2^{-\left(1/3\right)cl}\mbox{.}
\end{gather}

Summing with \eqref{eq:markov1}, we have that:
\begin{gather}
\Pr_{r}\left[\mbox{\ensuremath{r}\ is good}\right]\leq2\cdot 2^{-\left(1/3\right)cl}\mbox{.}
\end{gather}
\end{proof}

%% file: constant_rank.tex

\section{Algorithm for constant rank weighted bi-adjacency matrix $M$}\label{sec:Constant-rank-adjacency}

The previous sections show that the AIM problem for general bipartite graph $B$ and
social graph $G$ is hard to approximate to within any constant factor.
In this section, we restrict the (weighted) bi-adjacency matrix $\K$ between content provider nodes and consumer nodes to be of constant rank $r$,
and show that for this case we can obtain a constant factor approximation in polynomial
time. We denote this restricted problem {\sc AIM}-$r$. Our main algorithmic result is:

\begin{theorem} \label{thm:main theorem}
For any constant $r>0$ and $\delta, \varepsilon > 0$, {\sc AIM}-$r$ can be approximated to within $(1-1/\e-\varepsilon)^3$ 
	with probability $1-\delta$ and in time polynomial in $n, m, \lambda, 1/\varepsilon, \log (1/\delta)$, where $n=|U|$, $m=|V|$ and $\lambda$ is the maximum number of bits in 
	any entries of matrix $M$.
	\footnote{The running time is exponential in $r$. See Section~\ref{rem:submodular}
	for more details.}
\end{theorem}

For any fixed $X$, $\sigma(X,Y)$ is a monotone submodular function of $Y$.
Similarly, for any fixed $Y$, $\sigma(X,Y)$ is a monotone submodular function of $X$.
Each of those can be (approximately) optimized independently, 
thus the main algorithmic challenge is due to the interaction between the choice of $X$ and the choice of $Y$.
Intuitively, a constant rank bi-adjacency matrix creates an "information bottleneck" which restricts the complexity of this interaction.

How can we use the restriction on the matrix rank to optimize a non-linear objective?
To this end, we introduce in Subsection \ref{sub:relaxation} a relaxation of our objective function
which conveniently views $\K$ as a linear operator acting on $X$.
Because $\K$ has constant rank, the resulting subspace has a constant dimension;
in Subsection \ref{sub:net}, we show that we can efficiently (approximately) enumerate over all the points in this subspace.
Finally, given the (approximately) optimal choice of $X$, 
we can use standard submodular maximization techniques to (approximately) optimize over $Y$ (Subsection \ref{sec:alg}).

\subsection{Notation}

Henceforth, we use the following notational conventions.
For vector $\w x$, $x_i$ is the $i$-th element
of $\w x$. 
All vectors are column vectors (unless otherwise stated).
Let $|U|=n$ and $|V|=m$.
When the context is clear, we also use the index set $[n]=\{1,2,\ldots, n\}$ to represent
$U$ and index set $[m]=\{1,2,\ldots, m \}$ to represent $V$.

Given the provider seed set $X\subseteq U$ and the consumer seed set $Y\subseteq V$,
for convenience
we denote $\w{x}, \w{y}$ as the indicator vectors of $X$ and
$Y$, respectively.
An indicator vector for a subset $X$ of $U$ is a vector in $\{0,1\}^{n}$ such that
the entries corresponding to nodes in $X$ are $1$'s and nodes in $U\setminus X$ are $0$'s,
and indicator vector for subset $Y$ of $V$ is defined similarly.

Given $\w{x}$ and $\w{y}$, we use $f_j\left(\w{x},\w{y}\right)$ to denote the
{\em initial activation probability} of each node $j \in V$, which
is the probability that some node $i\in X$ activates $j$ based on matrix $M$,
i.e., $f_j(\w{x},\w{y})=y_j\left(1-\prod_{i\in [n]}\left(1- x_i\K_{ij}\right)\right)$.
We denote $\mathbf{f}(\mathbf{x},\mathbf{y})$ as the vector $(f_1\left(\w x,\w y\right),\ldots,f_{m}\left(\w x,\w y\right))^{\top}$.
Note that the initial activation of nodes in $Y$ from nodes in $X$ are mutually independent
for every node in $Y$.
Moreover, the initial activation probability of node $j\in V$
is not its {\em final activation probability},
which is the probability that node $j$
is activated by the end of the diffusion process, since $j$ may be later activated by
other nodes in $V$ through the diffusion process in the social graph $G$.
In particular, a node $j\in V\setminus Y$ has zero initial activation probability
by our model definition, but its final activation probability may be greater than zero.


\ignore{

\subsection{The multilinear relaxation}

In general, it is convenient to also consider fractional $\w x, \w y \in \[0,1\]^n$.
A standard technique in submodular optimization is the multilinear relaxation (e.g. \cite{CCPV07-submodular_multilinear}):
we define $F_j(\hat{\w x}, \hat{\w y})$ to be the expectation over $f_j(\hat{\w x}, \hat{\w y})$
where each element $x_i$ and $y_i$ is sampled independently from $\{0,1\}$
 with expectation $\hat{x}_i$ and $\hat{y}_i$, respectively.
Define $F(\hat{\w x}, \hat{\w y})$ analogously.

Similarly, 

Given a probability vector $\p \in [0,1]^m$ and an indicator vector
$\w y \in \{0,1\}^m$,
let $S_{\p, \w y}$ be the (random) set of nodes chosen in $Y$ according to vector $\p$.
That is, each node $j \in Y$ is included in $S_{\p,\w y}$ independently with
probability $\pi_j$, and no node in $V\setminus Y$ is included in $S_{\p,\w y}$.
We use $\f(\p,\w y)$ to denote the expected number
of activated nodes in $V$ when we first select seed set $S_{\p, \w y}$ and then let
the influence propagates in $G$ from $S_{\p, \w y}$.
Thus, according to our model defined in Section~\ref{sec:model}, we have
$\sigma(X,Y) = \f(\mathbf{f}(\w{x},\w{y}), \w y)$,
where $\w{x}$ and $\w{y}$ are the indicator vectors of sets $X$ and $Y$, respectively.

%

In the following two subsections, we prove some useful properties about function
$\f$ before proposing our Algorithm~\ref{alg:main} in Section~\ref{sec:alg}.
}

\subsection{A concave relaxation}
\label{sub:relaxation}

A key step in our algorithm is to approximate every coordinate $f_j(\mathbf{x},\mathbf{y})$
via the following concave relaxation\footnote{
	The relaxation is inspired by \cite{BPRSS}. Essentially
	the same relaxation was also used before by \cite{DRY11}
	in the context of Poisson rounding.}
\[
	F_j(\w{x},\w{y}) = y_{j}\left(1- \e^{-\left(\w x^\top M\right)_j }\right).
\]
Notice that this relaxation has two important features: (a) it is a function of
the linear form $\left(\w x^\top M\right)_j$, which allows us
to use the constant rank condition; and (b) it is both concave in
$\mathbf{x}$ for a fixed $\mathbf{y}$, and concave in $\mathbf{y}$
for a fixed $\mathbf{x}$ --- this will make it much easier to maximize
efficiently. Now we will show that it is a $\left(1-1/\e\right)$-approximation
of $f_j(\mathbf{x},\mathbf{y})$ by the following lemma.

\begin{lemma}
	\label{lemma:concave-approx}
	For any $\w x, \w y \in \{0,1\}^n$,
	\[
	\left(1-1/\e\right)f_j(\mathbf{x},\mathbf{y})\leq F_j(\mathbf{x},\mathbf{y})\leq f_j(\mathbf{x},\mathbf{y}), \forall j=1,\ldots, m.
	\]
\end{lemma}

\begin{proof}
	For the right inequality, since $\e^{-a} \geq 1 - a$ for any real $a$, we get
	\[
	F_j(\mathbf{x},\mathbf{y}) = y_j \left( 1-\e^{- \sum_{i\in [n]} x_i{\K}_{ij}} \right)
	\leq y_j \left( 1- \prod_{i\in [n]} (1 - x_i{\K}_{ij})\right) = f_j(\mathbf{x},\mathbf{y}).
	\]

	For the left inequality, because $(1-\e^{-1}) a \leq 1- \e^{-a}$ holds for any $a \in [0,1]$, we have
	\begin{align*}
		\left(1-\e^{-1}\right)\left(1-\prod_{i=1}^{n} \left(1 - x_i\K_{ij}\right)\right)
		&= \sum_{i=1}^{n} \left(\left(1-\e^{-1}\right) x_i\K_{ij} \right) \prod_{k=1}^{i-1} \left(1-x_k\K_{kj}\right)\\
		&\leq \sum_{i=1}^{n}\left(1-\e^{-x_i\K_{ij}}\right) \cdot \prod_{k=1}^{i-1} \e^{- x_k\K_{kj}}\\
		&= \sum_{i=1}^{n} \left(\exp\left(-\sum_{k=1}^{i-1} x_k\K_{kj}\right)- \exp\left(- \sum_{k=1}^{i} x_k\K_{kj}\right)\right)\\
		&=1-\e^{- \sum_{i \in [n]} x_i\K_{ij}}\\
		&=1-\e^{- \left(\w x^\top M\right)_j }.
	\end{align*}
	Multiplying $y_j$ on both sides of the above inequality, we get
	$\left(1-1/\e\right)f_j(\mathbf{x},\mathbf{y})\leq F_j(\mathbf{x},\mathbf{y})$.
\end{proof}

\subsection{Approximating initial activation probability $\w {f} (\w x,\w {y})$ via
	$(1+\varepsilon)$-net construction}
\label{sub:net}

Notice that the value of $F_j(\w x, \w y)$ is uniquely determined by $\w x^\top M$ and $\w y$.
We use $\Ima \K$ to denote the image of $\K$ when $\K$ is treated as a linear operator from $\{0,1\}^n$ to $\mathbb{R}^m$, 
i.e. $\Ima \K$ is the subspace of all vectors $\mathbf{x}^{\top}\K \in \mathbb{R}^m$, $\forall \w x \in \{0,1\}^n$.
Recall that since $\K$ has constant rank $r$, the dimension of $\Ima \K$ is also $r$.

Our next goal is to enumerate over (approximately) all feasible $\w s \in \Ima \K$.
For any $\varepsilon > 0$, we say that set $\S \subseteq \mathbb{R}^{\u}$ is a {\em multiplicative-$(1+\varepsilon)$-net} for $M$,
	if for every $\w x \in \{0,1\}^n$,
there exists a corresponding point
$\mathbf{s} = (s_1, s_2, \cdots, s_\u)^\top \in{\cal S}_\varepsilon$
such that for each coordinate $j \in \left[\u\right]$,
\begin{equation}
\label{eq:net}
	s_j  \leq \left(\mathbf{x}^{\top}\K\right)_j \leq s_j (1+\varepsilon).
\end{equation}
Henceforth we drop the multiplicative qualification and simply call such a set a $(1+\varepsilon)$-net.

\begin{lemma}\label{lemma:1-plus-eps-approx-2}
%
Let $M \in [0,1]^{n \times m}$ be a matrix with constant rank $r$, and whose entries can be represented with $\b$ bits.
Then for any error parameter $\varepsilon > 0$, 
we can output a polynomial-size $(1+\varepsilon)$-net for $M$ in time $\poly\left( n, m, 1/\varepsilon, \b \right)$.
\end{lemma}

\begin{proof}
	Below, we show how to construct a {\em weak $(1+\varepsilon)$-net} (in Algorithm \ref{alg:net}), 
	which instead of Equation \eqref{eq:net} gives the following weaker, 
	two-sided error guarantee:
\begin{equation}
\label{eq:weak-net}
	s_j  / (1+\varepsilon) \leq \left(\mathbf{x}^{\top}\K\right)_j \leq s_j (1+\varepsilon).
\end{equation}
	Given such an algorithm, one can construct a $(1+\varepsilon)$-net (with one-sided error)
	by constructing a weak $\sqrt{1+\varepsilon}$-net,
	and dividing every entry in the obtained weak $\sqrt{1+\varepsilon}$-net by  $\sqrt{1+\varepsilon}$.

	Since each entry of $M \in [0,1]^{n \times m}$ has at most $\b$ bits, then
	for any $\w{x} \in \{0,1\}^n$, every nonzero entry of $\w{x}^{\top}\K$
	is bounded in $[2^{-\b}, n]$.%
	\footnote{This step assumes of course that all the entries in $\K$ are positive.
		While this is a natural for an adjacency matrix of a social network,
		when taking the low-rank approximation of one, this may no longer
		be true. Nevertheless, a similar analysis continues to hold even when
		$\K$ has negative entries.%
	}
	Thus in each dimension, we lose no more than a $\left(1+\varepsilon\right)$-factor
	by only considering each position in
	$S_{\b}=\left\{ 0,2^{-\b},2^{-\b}\cdot\left(1+\varepsilon\right),2^{-\b}\cdot\left(1+\varepsilon\right)^{2},\dots,
	n \right\}$.
	
	We consider the partitioning $\left[0,n\right]^{\u}$ into hyper-rectangles that
	is induced by $\underbrace{S_{\b} \times \dots \times S_{\b}}_{\text{$m$ times}}$. More precisely, the set of hyper-rectangles $\mathcal{H}$
	is every possible direct product of intervals, i.e.,
	\begin{gather}
\label{eq:hyper-rectangles}
\mathcal{H} = \left\{
{ H = [a_1,b_1]\times [a_2,b_2] \times \cdots \times [a_m,b_m]:} \atop
{\forall i \in [m], \forall a_i \in S_\b \setminus \{n\}, b_i = \min\{\max\{(1+\varepsilon)a_i, 2^{-\b}\},n\}} 
\right\}. 
\end{gather}
	Observe that the disjoint union of those hyper-rectangles covers $\left[0,n\right]^{\u}$;
	in particular, for any $\w x\in\{0,1\}^n$, $\mathbf{x}^{\top}\K$ must belong to a hyper-rectangle $H \in \mathcal{H}$.  

	Notice that if $\mathbf{x}^{\top}\K$ and $\w s$ lie in the same hyper-rectangle $H$,
	then they satisfy the two-sided error approximation guarantee of Equation \eqref{eq:weak-net}.
	Our weak $(1+\varepsilon)$-net has at least one such $\w s$ for every $\mathbf{x}^{\top}\K$.

\ignore{
	Given any point $\mathbf{x}^{\top}\K$ inside a hyper-rectangle $H$,
	we can round down all its entries to the nearest value from $S_{\b}$.
	(This corresponds to $\w a = \left(a_1,\dots , a_m\right)$ in Equation \eqref{eq:hyper-rectangles}.)
	We henceforth denote this point $\w a_H$ and call it the {\em bottom corner} of $H$.
	Notice that $\w a_H$ is an approximation of $\mathbf{x}^{\top}\K$ in the sense of Inequality~\eqref{eq:net}.
	Our $(1+\varepsilon)$-net will consist of all $\w a_H$ for hyper-rectangles $H$ that have a non-empty intersection with $\Ima \K$.
}

	Consider any hyper-rectangle $H \in \cal H$ with a non-empty intersection with  $\Ima \K$,
	and let $I_H = H \cap \Ima \K$ denote their intersection.
	Since $I_H$ is an intersection of convex polytopes, it is also a convex polytope,
	defined by $m-r$ linearly independent equations that define $\Ima \K$
	and $2m$ inequalities that define $H$. 
	Every vertex $\w v$ of this polytope must lie on the intersection of $m$ linearly independent constraints: 
	$m-r$ equations and $r$ inequalities.
	In other words, $\w v$ lies in the intersection of at least $r$ facets of $H$;
	i.e. there exist $i_1, \dots , i_r \in [m]$ such that $v_{i_k} \in \{a_{i_k}, b_{i_k}\} \;\;\forall k\in [r]$.
	Furthermore, these $r$ coordinates must correspond to linearly independent columns of $\K$.

%
	Therefore, the following polynomial-time procedure (summarized in Algorithm \ref{alg:net}) 
	correctly construct a weak $(1+\varepsilon)$-net:
	First we take a submatrix $M'$ of $M$ with $r$ linearly independent rows, so $\Ima M'=\Ima M$ (Line~\ref{alg:submatrix}). Next
	enumerate over all $m\choose r$ $r$-tuples of linearly independent columns of $\K'$ 
	(Line \ref{alg:net-line-for-1}). 
	For each $r$-tuple, enumerate over all vectors in $\underbrace{S_{\b} \times \dots \times S_{\b}}_{\text{$r$ times}}$
	(Line \ref{alg:net-line-for-2}). 
	Each such vector uniquely defines a point in $\w s \in H \cap \Ima \K$
	(Line \ref{alg:net-Gaussian}).
	Finally, adding it to $\S$ guarantees that we can approximate any other $\w x^{\top}\K  \in H$
	(Line \ref{alg:net-add-s}). 
\end{proof}

\begin{algorithm}[t]
	\caption{Construct a weak $\left(1+\varepsilon\right)$-net over the image of $\K$}
	\label{alg:net}
	\LinesNumbered
	
	\KwIn{The matrix $\K$ of rank $r$, and the error factor $\varepsilon$}
	\KwOut{The multiplicative $\left(1+\varepsilon\right)$-net \S.}

	Let $\b$ be the maximum number of bits in any entry of $\K$, and
		$S_{\b} = \left\{ 0, 	2^{-\b}, 
							2^{-\b}\cdot\left(1+\varepsilon\right), 
							2^{-\b}\cdot\left(1+\varepsilon\right)^{2}, \dots, n \right\}$ 
\\
	$\S \gets \emptyset$\\
	Let $M'$ be the $r\times m$ submatrix of $M$ with $r$ independent row vectors $\w v_1,\ldots,\w v_r \in [0,1]^m$ \label{alg:submatrix}\\

	\For{$\left( i_1,i_2,\ldots, i_r \in [m] \text{ s.t. } i_1<i_2<\cdots<i_r \atop \text{and the corresponding columns of $\K'$ are independent} \right) $} 
	{	\label{alg:net-line-for-1}
		\tcc{Enumerate every $r$ coordinates according to the grid}
		\For{$k_1,\ldots, k_r \in S_\b$}{	\label{alg:net-line-for-2}
			Construct linear system of $r$ equations: 
			$\left\{
			\begin{aligned}
				(\hat z_1\w v_1 + \cdots + \hat z_r \w v_r)_{i_1} & = & k_1 \\
				(\hat z_1\w v_1 + \cdots + \hat z_r \w v_r)_{i_2} & = & k_2 \\
				&\vdots&\\
				(\hat z_1\w v_1 + \cdots + \hat z_r \w v_r)_{i_r} & = & k_r
			\end{aligned}
			\right.
			$\\
			\label{alg:net-line-for-3}
			\tcc{ $\hat z_1\w v_1 + \cdots + \hat z_r \w v_r$ is an $m$-dimensional vector, and denote
				$(\hat z_1\w v_1 + \cdots + \hat z_r \w v_r)_i$ as its $i$-th coordinate for each $i \in [m]$ }
			
			Use Gaussian elimination to derive the solution $(\hat z_1', \ldots, \hat z_r')^\top \in \mathbb{R}^r$ 
			\label{alg:net-Gaussian}				\\
			\label{alg:net-add-s} 
			$\w s \gets \hat z_1'\w v_1 + \cdots + \hat z_r' \w v_r$;
			$ \S \gets \S \cup \{\w s\}$
			\\
\ignore{			\For{ $j \in [m]$} 
			{	$s_j \gets \max\{ t \in S_{\b} \wedge t \le s'_j\}$ \tcc{Find the bottom corner} \label{alg:net-bottom} }}
			
		}
	}
	\KwRet{ \S }
\end{algorithm}

We define 
\begin{equation}
	\hat F_j(\w s, \w y) = y_{j}\left(1- \e^{-s_j }\right), \label{eq:fhat}
\end{equation}
thus $\hat F_j(\w s, \w y) = F_j(\w x, \w y)$ for all $\w s = \w x^\top M$.
Notice that
 the definition of $\hat F_j(\w s, \w y)$ naturally extends also to $\w s \notin \Ima \K$.
In the following lemma we relate $\hat F_j$ to the original $f_j$.


\begin{lemma}\label{lemma:approximation-vector}
Let $\varepsilon>0$, and let $\S$ be a $(1+\varepsilon)$-net for $\K$.
Then for every $\w x \in \{0,1\}^n$, there exists an $\w{s} \in \S$
such that for every $\w y\in \{0,1\}^m$ and $j \in [m]$, 

	\begin{align*}
		& (1-1/\e-\varepsilon) f_j(\w{x},\w{y}) \le \hF_j(\w s,\w y) \le f_j(\w{x},\w {y}).
	\end{align*}
\end{lemma}

\begin{proof}
	By Lemma \ref{lemma:1-plus-eps-approx-2}, for any $\w x \in \{0,1\}^n$, we know that there exists $\w{s} \in \S$ satisfying Inequality~\eqref{eq:net}. 
	Moreover, for any $\w{y}$, we claim that
	\begin{align}	
	\hat{F}_j(\w{s},\w{y})
		\geq F_j(\w{x},\w{y}) / \left(1+\varepsilon\right)
		\geq \left(1-1/\e-\varepsilon\right)f_j\left(\mathbf{x},\mathbf{y}\right),
		\label{eq:lower-bound-of-F-with-optimal-value}
	\end{align}
	where the second inequality is due to Lemma~\ref{lemma:concave-approx} and the fact
	that $(1-1/\e)/(1+\varepsilon) \ge (1-1/\e -\varepsilon)$ for all $\varepsilon >0$.
	Recall that $ F_j(\mathbf{x},\mathbf{y}) = 1 - \e^{- y_j \cdot \left(\w x^\top M\right)_j } $,
	thus to show the first inequality of Eq.~\eqref{eq:lower-bound-of-F-with-optimal-value},
	we only need to show that,
	\[
		1-\e^{-s_{j} y_j}
			\geq \left(1 - \e^{-y_j \cdot \left(\w x^\top M\right)_j } \right) / \left(1+\varepsilon\right).
	\]
	For $y_j = 0$, it is trivial. For $y_j = 1$, since $(1+\varepsilon) s_{j} \geq \left(\w x^\top M\right)_j$, we have
	$\e^{- (1+\varepsilon) s_{j} } \leq \e^{- \left(\w x^\top M\right)_j }$.
	Therefore, it is enough to show that
	$
	\left(1-\e^{-s_{j} }\right)
		\geq \left(1-\e^{- (1+\varepsilon) s_{j} }\right) /\left(1+\varepsilon\right),
	$
	namely,
	\[
		\frac{ \varepsilon+ \e^{-(1+\varepsilon) s_{j}} }{1+\varepsilon} \geq \e^{- s_{j}}.
	\]
	Note that the above inequality holds due to Weighted AM-GM inequality $\frac{ax+by}{a+b} \geq x^{\frac{a}{a+b}}y^{\frac{b}{a+b}}$,
	$\forall x, y, a, b \in \mathbb{R}^{+}$ by letting $a=\varepsilon, x=1,b=1,y=\e^{-(1+\varepsilon) s_j}$.

	Furthermore, we claim that
	\[
	f_j\left(\w x,\w y \right)
		\geq F_j\left(\w x,\w y \right)
		\geq \hF_j\left(\w s, \w y\right),
	\]
	where the first inequality is due to Lemma~\ref{lemma:concave-approx}, and the second one follows from Equation~\ref{eq:net}.
\ignore{
Therefore we only need to prove 
	$ F_j\left(\w x,\w y \right)
	\geq \frac{1}{1+\varepsilon}\hF_j\left(\w s, \w y\right)$.
	Similarly, it is enough to show that, for each $j$,
	\[
	1 - \e^{-y_j \cdot \left(\w x^\top M\right)_j }
		\geq \left(1-\e^{-s_{j} y_j}\right) / \left(1+\varepsilon\right).
	\]
	The above holds trivially for $y_j = 0$, while for $y_j=1$ it can be derived by following $\left(\w x^\top M\right)_j \geq s_{j} / \left(1+\varepsilon \right)$
	and the symmetric analysis. Therefore, by concatenating inequalities we can conclude that, for any $\w y\in\{0,1\}^m$ and any $j \in [m]$,
	\[ 
		f_j(\w{x},\w {y}) \geq \frac{1}{1+\varepsilon} \hF_j(\w s,\w y). 
	\]
}
\end{proof}

\ignore{

\subsection{Properties of function $\f(\cdot,\cdot)$}
%

In this subsection, we prove that function $\f(\p, \w y)$, defined on IC model,  has four properties in Lemma \ref{lemma:properties}. 



Let $\w 0$ and $\w 1$ denote the vectors of all zeros and ones, respectively.
As a convention, we treat any vector function $g(\w a)$ where $\w a$ is an
indicator vector of some set $A$ as the same as a set function $g(A)$.
Thus, when we say that a vector function $g(\w a)$ is monotone and
submodular on the indicator vector $\w a$, 
we mean that $g(A) \triangleq g(\w a)$ is monotone and submodular on $A$.
Moreover,  for $\p,\p' \in [0,1]^m$, we define $ \p\leq \p'$ if $\forall i\in[m]: \pi_i\leq \pi'_i$.
We say that a vector function $g:[0,1]^m\rightarrow R$ is monotone if for any  $\p,\p'\in[0,1]^m$ 
	such that $\p\leq \p'$ we have that $g(\p)\leq g(\p')$.\footnote{In this paper, "monotone"
	always means monotonically non-decreasing.}
For convenience, we write $ \p \odot  \w y \triangleq (\pi_1 y_1,\ldots,\pi_m y_m)^T$ to denote the operation of multipling vectors $\w y= (y_1,\ldots, y_m)^T$ and $\p=(\pi_1,\ldots,\pi_m)^T$  piece-wise on every
	coordinate.

\begin{lemma}\label{lemma:properties}
The function  $\f(\p, \w y)$ has the following four properties:
	\begin{enumerate}
		\item \label{clm:monotone} 
		For any fixed $\w y\in \{0,1\}^m$, $\f(\p,\w y)$ is a monotone function on $\p$.
		
		\item \label{clm:submodular-for-u}
		For any fixed 
		$\w{y}\in\{0,1\}^{m}$, 
		$\f(\w f(\w x, \w y),\w y)$ is a monotone and submodular function on $\w x$.
		
		\item \label{clm:submodular-for-v}
		For any fixed $\p \in [0,1]^m$, $\f( \p \odot \w y,\w y) $ is a monotone and submodular function on $\w y$.
		
		%
		
		\item \label{clm:linear}
		 For any constant $\alpha\in[0,1]$, $\f(\alpha \p, \w y)\geq\alpha \f(\p, \w y)$.

	\end{enumerate}
\end{lemma}

\begin{proof}[(Sketch)]
	Denote $\U(Y) = \f(\w 1, \w y)$ where $Y$ is the same subset of indicator vector $\w y$.
	Let $\w s_{\p,\w y}$ denote the indicator vector of the (random) set
	$S_{\p,\w y}\subseteq Y$ where each element in $Y$ is sampled independently according to the 
	probability vector $\p$, and no element in $V\setminus Y$ is sampled.
	Thus
	\begin{equation}
		\f(\p, \w y)=\E {\f( \w 1,\w s_{\p,\w y})}=\E{\U(S_{\p,\w y})}, \label{eq:formulataion-of-f}
	\end{equation}
	where the expectation is taken from all random selections of set $S_{\p,\w y}$.

	\paragraph{Property~(1)} It is straightforward to see the function $\f(\p,\w y)$ is monotone on $\p$.
	
	\paragraph{Property~(2) \& (3)} 	
	Note that, by \cite{kempe2003maximizing}, it is proved that for the IC model, the resulting influence function $\f(\w 1, \w y)$ is non-negative,
	monotone and submodular on $\w y$. Thus to prove property (2), (3), we only need to show the functions $\f(\w f(\w x, \w y),\w y)$, $\f( \p \odot \w y,\w y)$ are resulting influence function of some IC models on $\w x, \w y$, respectively. 
	
	For function $\f(\w f(\w x,\w y), \w y)$, convert the social graph $G$ to new graph $G'$ by deleting all edges from the nodes in $U$ to the nodes in the set $V\setminus Y$. Then for fixed $\w y \in \{0,1\}^m$, $\f(\w f(\w x,\w y), \w y)$ is the influence function of $G'$ of selecting nodes in $U$ according to $\w x$, which is the influence function of standard IC model.
	
	For function $\f( \p \odot \w y,\w y) $,  construct graph $G''$ by adding the mirror node $v_i'$ according to each node $v_i$ in $G$ and connecting $v_i',v_i$ with an edge with probability $\pi_i$. It is easy to see that the function $\f( \p \odot \w y,\w y) $ computes the resulting influence  of $G''$ after choose mirror nodes according to $\w y$, which is also the influence function of  standard IC model.

	\paragraph{Property~(4)}
	
	As is implied by Eq.~\eqref{eq:formulataion-of-f},
	we only need to show that
	\begin{align}
		\E{\U(S_{\alpha \p, \w y })}\geq \alpha \E{\U(S_{\p, \w y})}. \label{eq:properties-3-to-prove}
	\end{align}
	
	Consider the random process that sample the set $S_{\alpha \p, \w y}$ according to the probability vector $\alpha \p$.
	It is equivalent to the sampling of the set $S_{\p, \w y}$ from $Y$ according to
	$\p$, followed by another sampling of its subset $S_{\alpha \w 1, \w s_{\p, \w y}}$ from $S_{\p, \w y}$
	according to $\alpha \w 1$ independently. That is,
	$$
	\E{\U(S_{\alpha \p, \w y })} = \E{\E{\U(S_{\alpha \w 1, \w s_{\p,\w y}})
			\mid S_{\p, \w y} }},
	$$
	where the inner expectation on the right-hand side is taken over random vector
	$\w s_{\alpha \w 1, \w s_{\p,\w y}}$ for any fixed $\w s_{\p,\w y}$,
	and the outer expectation is taken over
	random vector $\w s_{\p, \w y }$.
	
	Thus it is enough to prove that, for any fixed set $A \subseteq Y$ with indicator
	vector $\w a$,
	\begin{equation}
		\E{\U(S_{\alpha \w 1, \w s_{\p,\w y}})
			\mid S_{\p, \w y} = A} = \E{\U(S_{\alpha \w 1, \w a})}
		\geq \alpha \U(A). \label{eq:submodularconvex}
	\end{equation}

	According to Lemma 2.2 of \cite{vondrak2007submodularity}, we have for any full set $\Omega$ and
	any submodular function $\phi$ on $\Omega$, if $\w s_{\alpha \w 1, \w 1 }$
	denotes the random indicator vector for $\Omega$ where each element of $\Omega$ is
	selected independently with probability $\alpha$, then
	$$\E{\U(S_{\alpha \w 1, \w 1 })}\geq \alpha \U(\Omega)+(1-\alpha) \U(\emptyset).$$
	Inequality~\eqref{eq:submodularconvex} is immediate from above when we treat set $A$
	as the full set $\Omega$, and the fact that $\phi(\emptyset) \ge 0$.
\end{proof}


%
}

Analogously to $\hat F_j(\w s, \w y)$, we can also define $\hat \sigma(\w s, \w y)$ 
to be the expected number of (eventually) activated consumer nodes 
given that the set of {\em initially} activated consumer nodes is distributed according to $\hat{\w F}(\w s, \w y)$.
(I.e. each node $j \in V$ is independently initially activated with probability $\hat F_j(\w s, \w y)$.)
For any fixed $\w s$, by the definition of $\hat F_j(\w s, \w y)$ in Eq.~\eqref{eq:fhat}, we know that
	$\hat \sigma(\w s, \w y)$ is equivalent to the influence spread obtained by selecting seed set $Y$
	(indiated by $\w y$),  each seed $j\in Y$ being activated with probability $1- \e^{-s_j }$ and then
	influence propagated in the social graph $G$.
Then, by the result in \cite{kempe2003maximizing}, it is straightforward to see that $\hat \sigma(\w s, \w y)$ is monotone and submodular on $\w y$. \footnote
{For fixed $\w s$, when we say that a vector function $\hat \sigma(\w s, \w y)$ is monotone and
	submodular on the indicator vector $\w y$,
	we mean that $\hat \sigma (\w s, Y) \triangleq \hat \sigma(\w s, \w y)$ is monotone and submodular on 
		set $Y$ indicated by $\w y$.}
	

Furthermore, our $(1+\varepsilon)$-net continues to (approximately) capture all possible inputs to $\hat \sigma$:

\begin{lemma}\label{lemma:approximation-sigma}
Let $\varepsilon>0$, and let $\S$ be a $(1+\varepsilon)$-net for $\K$.
Then for every $\w x \in \{0,1\}^n$, there exists an $\w{s} \in \S$
such that for every $\w y\in \{0,1\}^m$, 
	\begin{align*}
		& (1-1/\e-\varepsilon) \sigma(\w{x},\w{y}) \le \hat \sigma(\w s,\w y) \le \sigma(\w{x},\w {y}).
	\end{align*}
\end{lemma}

\begin{proof}
The right inequality follows immediately from Lemma \ref{lemma:approximation-vector}.

To prove the left inequality, we unfortunately need to define yet another function.
For $\w z \in \{0,1\}^m$, let $\rho(\w z)$ be the expected number of activated nodes 
given that the set of {\em initially} activated nodes is $Z \subseteq Y$ indicated by $\w z$.
Notice that $\rho(\w z)$ is also submodular \cite{kempe2003maximizing}.

For a fractional $\bar{\w z} \in [0,1]^m$, we extend $\bar \rho(\bar{\w z})$ to be the expectation 
over integral $\w z \in \{0,1\}^m$, where each coordinate is sampled independently with expectation $\bar z_j$.
(I.e. $\bar \rho(\bar{\w z})=\sum_{\w z\in\{0,1\}^m} \Big[\rho(\w z) \cdot\prod\left(\bar{z}_i ^ {z_i}\cdot (1-\bar z_i)^{1-z_i}\right) \Big]  $.)
Thus, $\sigma(\w{x},\w{y}) = \bar \rho\left( \w f(\w x, \w y) \right)$ 
and $\hat \sigma(\w{s},\w{y}) = \bar \rho\left( \hat{\w F}(\w s, \w y) \right)$.

Finally, since $\bar \rho(\w z)$ is the multilinear extension of a submodular function, 
we have (e.g. by Lemma 2.2 of \cite{vondrak2007submodularity}) that
	\begin{align*}
		& (1-1/\e-\varepsilon) \bar \rho\left( \w f(\w x, \w y) \right) \le \bar \rho\Big( (1-1/\e-\varepsilon) \cdot \w f(\w x, \w y) \Big) \le \bar \rho\left( \hat{\w F}(\w s, \w y) \right).
	\end{align*}
\end{proof}

\subsection{From $(1+\varepsilon)$-net to approximation algorithm}\label{sec:alg}

Armed with our $(1+\varepsilon)$-net, we can use standard submodular maximization techniques to approximately solve {\sc AIM}-$r$.
The full algorithm is summarized in Algorithm \ref{alg:main}, and referred as
	Sampled Double Greedy ({\sf SDG}) algorithm.

For every $\w s \in \S$, let $\w y_{\w s}$ be a $(1-1/\e-\varepsilon)$-approximation to the feasible vector $\w y_{\w s}^*$ that maximizes $\hat \sigma(\w s,\w y)$.
Recall that such an approximation can be found in polynomial time (with high probability) via standard submodular maximization techniques (see e.g. \cite{kempe2003maximizing} as well as Subsection \ref{rem:submodular}).
Let $(\w x^* , \w y^*)$ be the optimal feasible solution to the {\sc AIM}-$r$ problem.
Then for some $\w s^*$, we have that for every $\w y \in \{0,1\}^m$:
\begin{gather*}
	 (1-1/\e-\varepsilon) \sigma(\w{x}^*,\w{y}) \le \hat \sigma(\w s^*,\w y) \le \sigma(\w{x}^*,\w {y});
\end{gather*}

and in particular, 
\begin{gather}
\label{eq:best-y}
	\hat \sigma(\w s^*,\w y_{s^*}) \ge (1-1/\e-\varepsilon) \hat \sigma(\w{s}^*,\w y^*) \ge (1-1/\e-\varepsilon)^2 \sigma(\w{x}^*,\w{y}^*).
\end{gather}

For each fixed $\w y_{\w s}$, function $\sigma(\w x,\w y_{\w s})$ is monotone and
	submodular on $\w x$.
This is because we can remove all edges from $U$ to nodes not in the set indicated by $\w y_{\w s}$,
	and $\sigma(\w x,\w y_{\w s})$ is the influence spread of seed set $X$ in the combined bipartite
	graph and social graph after removing those edges.
Since diffusion from $X$ in this subgraph can be viewed as IC model diffusion, by
	\cite{kempe2003maximizing} we know that $\sigma(\w x,\w y_{\w s})$ is monotone and submodular on $\w x$.
Then we can take $\w x_{\w s}$ to be a $(1-1/\e-\varepsilon)$-approximation to the feasible vector $\w x_{\w s}^*$ that maximizes $\sigma(\w x,\w y_{\w s})$.
Thus, for $\w y_{s^*}$, we have
\begin{gather}
\label{eq:best-x}
	\sigma(\w x_{s^*},\w y_{s^*}) \ge (1-1/\e-\varepsilon) \sigma(\w{x}^*,\w y_{s^*}) \ge (1-1/\e-\varepsilon) \hat \sigma(\w{s}^*,\w y_{s^*}).
\end{gather}
Taking \eqref{eq:best-y} and \eqref{eq:best-x} together, we have
\begin{gather*}
	\sigma(\w x_{s^*},\w y_{s^*}) \ge (1-1/\e-\varepsilon)^3 \sigma(\w{x}^*,\w{y}^*).
\end{gather*}
	
This completes the proof of Theorem \ref{thm:main theorem}.
\qed

\begin{algorithm}[t]
	\caption{{\sf SDG}: A constant-factor approximation to {\sc AIM}-$r$ 
	}
	\label{alg:main}
	\LinesNumbered
	
	\KwIn{Bi-adjacency matrix $\K$, 
	adjacency matrix $P$, budgets $b_1$, $b_2$, accuracy parameter $\varepsilon > 0$}
	\KwOut{Subsets $(X,Y)$ that approximate the optimal $\sigma(X,Y)$}
	Construct $(1+\varepsilon)$-net $\S$ from $M$ by Algorithm~\ref{alg:net}\\
	\For{$\w s\in \S$}{
		Use greedy algorithm to find solution $\w y_{\w s}$ on submodular function
		 $\hat \sigma(\w s, \cdot)$ with budget $b_2$\\ \label{alg:sub-for-y} 
		Use greedy algorithm to find solution $\w x_{\w y_{\w s}}$ on submodular function  $\sigma(\cdot, \w y_\w s)$ with budget $b_1$\\ \label{alg:sub-for-x}
	}
	\KwRet{ $\argmax_{(\w x_{\w y_\w s}, \w y_\w s): \w s\in \S} \sigma(\w x_{\w y_\w s}, \w y_\w s)$ } \label{alg:findbests}
	
\end{algorithm}


\subsection{Running time}
\label{rem:submodular}

We have completely ignored the question of how to compute $\sigma(\w x, \w y)$ and its relaxation  $\hat \sigma(\w s, \w y)$.
In particular, this computation is necessary for the submodular maximization procedures used in Algorithm \ref{alg:main}.
Although their exact computations are \#P-hard \cite{chen2010scalable}, 
	both can be efficiently approximated with arbitrarily good precision
	by sampling from the corresponding random processes.
In particular, for the independent cascade model, both greedy steps in Algorithm \ref{alg:main} can apply the
	near-linear time algorithms in~\cite{borgs2014maximizing,TXS14}.

The total running time of {\sf SDG} is
	$O(m^r (\b+\log n)^r (b_1+b_2)(n+m+\ell) \varepsilon^{-(r+2)} \log(1/\delta)))$,
	where $O(m^r (\b+\log n)^r \varepsilon^{-r})$ is the size of $(1+\varepsilon)$-net $S_\varepsilon$,
	and $O(\varepsilon^{-2}(b_1+b_2)(n+m+\ell) \log(1/\delta))$ is the time running two greedy steps
	using the algorithm in \cite{TXS14}, with $\ell$ being the total number of edges in $B$ and $G$.
In some situations $r$ could be very small, e.g. $r=1$ when the influence probability from provider $i$
	to consumer $j$ can be approximated as the product of provider $i$'s influence strength and 
	consumer $j$'s susceptibility.
Moreover, in practice seed consumers may be selected from a candidate set of size $m'$
	(e.g. the fan base of a product) much smaller than the social network size ($m' << m$), 
	then the dominant term $m^r$ would be replaced by the much smaller $(m')^r$.
Therefore, {\sf SDG} would be efficient in these practical situations.

%
%
%
%

\section{General diffusion model}\label{sec:general}
Our results can be generalized to support any diffusion model in the social network $G$
	that has a monotone and submodular influence spread function.
More specifically, for any subset $Z \subseteq V$ in the social graph $G$, 
	let $\rho(Z)$ be the influence spread of $Z$ in $G$, that is, the expected number of activated
	nodes in $G$ after the diffusion process when $Z$ is selected as the initial seed set.
The diffusion model in the combined bipartite graph $B$ and social graph $G$ with selected
	seed providers $X$ and seed consumers $Y$ is as follows:
First, all seed providers in $X$ are activated; and then following the probabilities given in the
	bi-adjacency matrix $M$, a subset of seed consumers, $Z \subseteq Y$, is activated
	(each $j\in Y$ is activated independently with probability $\w f(\w x, \w y)$ as before).
Then the diffusion from $Z$ follows the social graph diffusion model, with expected spread
	$\rho(Z)$.
We still use notation $\sigma(X,Y)$ to represent the influence spread of $X$ and $Y$ in the combined
	network.
Then we have $\sigma(X,Y)=\sum_{Z\subseteq Y} \Pr_{X}[Z]\rho(Z)$, where
	$\Pr_{X}(Z)$ is the probability that  $Z$  is the initially activated set in $Y$
	by provider seed set $X$ according to the matrix $M$.

One particular instantiation of interest is that of {\em background propagation}:
	each consumer has some initial (``background'') probability of being influenced
	by content providers even without being seeded by the advertiser;
	if the advertiser seeds the consumer, she has a (higher)
	{\em boosted probability} of being influenced by the content providers.
To implement this model as a submodular influence function, 
	we can sample the result of the diffusion process from the background probabilities
	and incorporate the expected outcome into the definition of $\rho(\cdot)$.

In general, we can show that as long as $\rho(\cdot)$ is monotone, submodular, and polynomial-time
	computable, and matrix $M$ is of constant rank as assumed before,
	AIM problem is still solvable in polynomial time.
The main revision of the proof is to show that in the general model
	$\sigma(X,Y)$ is still monotone and submodular when we fix either $X$ or $Y$.	
\ifFULL
See Appendix \ref{appendix:sec:definition}
\else
See the full report \cite{CLLR15}
\fi
for details.

\section{Conclusion}

In this paper we propose the amphibious influence maximization (AIM) model 
as a proxy framework that combines traditional marketing
	via content providers together with viral marketing to consumers 
	in social networks.
We show that the associated computational problem is \NP-hard to approximate to any constant factor, and provide a
	polynomial-time algorithm with $(1-1/\e-\varepsilon)^3$ approximation ratio for any (polynomially small) $\varepsilon > 0$ when we restrict the weighted bi-adjacency matrix $M$ for the provider-consumer network
	to be of constant rank.

It would be interesting to see to what extent {\em amphibious marketing}
	(i.e. targeting individual users via a combination of traditional content providers and social network viral marketing) can
	be implemented in practice.
Beyond the algorithmic challenge of optimizing the sets of seed providers and consumers we discuss in this paper,
	this notion raises many interesting challenges in terms of learning the influence factors (the adjacency matrices in our model),
	privacy of the consumers, economic incentives, etc.

From the perspective of theoretical computer science,
	we view our algorithm for AIM with low rank assumption 
	as part of the ongoing effort in the community to incorporate
	assumptions that are both reasonable in practice, and allow better algorithmic results.
In this context we remark that although our low rank assumption is most natural in the context of content providers-consumers influence matrix,
	it is also closely related to another important property that has been observed in graphs of social networks:
	the eigenvalues exhibit a power law \cite{FFF99-power_law, MP02-eigenv_power_law}.

\ignore{
Future studies on AIM may investigate whether the approximation ratio or the running time of
	the current algorithm could be improved, or if there is other natural assumptions on the network
	besides the constant rank assumption on $M$ that could allow polynomial time approximation algorithms.
Other variants of AIM may also be worth further study, such as what if the influence probability between
	two selected seed consumers is also boosted, or what if the probabilities from non-seed providers to seed
	consumers and seed providers to non-seed consumers are boosted to different levels.
}

%% file: appendix-definition.tex
\appendix
\newpage
\section*{APPENDIX}

\renewcommand{\f}{\hat{\sigma}}

\section{Generalized model}
\label{appendix:sec:definition}

In this appendix, we extend the underlying diffusion model in social graph $G$  to allow a general monotone and submodular influence spread function $\rho(\cdot)$. 
We  show that as long as the general influence spread function $\rho(\cdot)$ can be approximated in polynomial time and matrix $M$ is of constant rank as assumed before, 
The same SDG Algorithm (except that we now need a computation oracle for $\rho(\cdot)$, see
	Algorithm~\ref{appendix:alg:main}) solves the generalized AIM problem with the same approximation ratio in polynomial time.
We also discuss a particular consequence of this generalization that allows
	each consumer node to have a background activation probability even if
	it is not selected as a seed.


\subsection{Definition of the generalized model}
\label{sec:generalmodel}


Instead of assuming the particular IC model, we  assume that the influence spread function over the social graph  is a  general monotone and submodular function.
Formally, we use $G=(V,\rho)$ to denote the social graph where $\rho$ is the general monotone and submodular function computing the resulting influence spread of seed consumers. Namely,  given any subset $Z\subseteq V$, $\rho(Z)$ is the resulting influence spread through $G$ when $Z$ is the set of consumers initially influenced by the content providers. In the following, we assume that $\rho(Z)$ is monotone and submodular.   

We still use notation $\sigma(X,Y)$ to represent the influence spread of $X$ and $Y$ in the combined
network.
Then we have $\sigma(X,Y)=\sum_{Z\subseteq Y} \Pr_{X}[Z]\rho(Z)$, where
$\Pr_{X}(Z)$ is the probability that  $Z$  is the initially activated set in $Y$
by provider seed set $X$ according to the matrix $M$.

Our goal is still the same as the original  AIM: to find a set $X$ of seed providers of size $b_1$ and a set $Y$ of seed consumers of size
$b_2$ such that they work together to generate the largest influence spread, namely 
	maximizing $\sigma(X,Y)$. 


%% file: appendix-boost.tex

\subsection{Result and Proof for the Extension}
\label{appendix:sec:alg}
We still restrict the bi-adjacency matrix $M$ to be of constant rank $r$. 
Moreover, we  assume that there is a value oracle $\cal O$ computing  the general influence spread function $\rho(Z)$ for  any set $Z\subseteq V$ with running  time $t_\cal O$.  
The full algorithm is summarized in Algorithm \ref{appendix:alg:main}, 
	and the only adaptation is to use the value oracle $\cal O$.
In particular, when we use greedy algorithm on function 
 	$\f(\w s, \cdot), \sigma(\cdot, \w y _ {\w s})$, we can use Monte Carlo simulation to obtain
 	the initially activated set $Z$ in $G$ and then obtain $\rho(Z)$ from the oracle,
 	which would give us a $1 \pm \varepsilon$ approximation of $\f(\w s, \w y), \sigma(\w x, \w y _ {\w s})$ 
 	with high probability, for any fixed $\w s \in [0,\infty)^m, \w x, \w y,\w y _ {\w s} \in \{0,1\}^m$. 
 The theoretical guarantee is stated as follows: 

%
%
%
%
%

\begin{theorem} \label{appendix:thm:main theorem}
	Assuming that there exists a value oracle $\mathcal{O}$ computing the function $\rho(\w y)$
	for  any indicator vector $\w y$ with running time $t_\mathcal{O}$,
	for any $\delta,\varepsilon>0$, with probability $1-\delta$, 
	Algorithm~\ref{appendix:alg:main} solves the generalized {\sc AIM} problem with
	constant rank-$r$ matrix $\K$
	with approximation ratio $(1-1/\e-\varepsilon)^3$  and in time
	polynomial in $n, m, \b, 1/\varepsilon, \log (1/\delta), t_\mathcal{O}$.
%
\end{theorem}

\begin{algorithm}[t]
	\caption{ {\sf SDG}: A constant-factor approximation for generalized AIM-$r$}
	\label{appendix:alg:main}
	\LinesNumbered
	\KwIn{Bi-adjacency matrix $\K$, value oracle $\mathcal{O}$ that calculates $\rho(\w y)$,
		budget $b_1$, $b_2$, and parameter $\varepsilon$}
	\KwOut{Subsets $(X,Y)$ that can make $\sigma(X,Y)$ be near-optimal value.}
	Construct $(1+\varepsilon)$-net $\S$ from $M$ by Algorithm~\ref{alg:net}\\
	\For{$\w s\in \S$}{
	Use greedy algorithm with oracle $\cal O$ to find $\w y_{\w s}$ on submodular function
		$\f(\w s, \cdot)$ with budget $b_2$  \label{alg:sub-for-y}	\\
		Use greedy algorithm with oracle $\cal O$ to find $\w x_{\w y_{\w s}}$ on submodular function $\sigma(\cdot, \w y _ {\w s})$ with budget $b_1$\\ \label{alg:sub-for-x}
			
	}
	\KwRet{ $\argmax_{(\w x_{\w y_\w s}, \w y_\w s): \w s\in \S} \sigma( \w x_{\w y_\w s}, \w y_\w s)$ } \label{alg:findbests}
\end{algorithm}

Notice that the proof of above theorem is essentially the same as Theorem \ref{thm:main theorem}, except that  we have to prove, for the general model, the function $\sigma(\cdot, \w y)$  and its relaxation $\f(\w s, \cdot)$  are still monotone and submodular. Thus it is enough for us to prove the two submodularities which are  stated 
	in the following lemma.

%% file: appendix-submodularity.tex
\begin{lemma} \label{appendix:lemma:properties}
We have the following two properties:
	\begin{enumerate}
	
		\item \label{appendix:clm:submodular-for-u} 
		For any fixed 
		$\w{y}\in\{0,1\}^{m}$, 
		$\sigma(\w x,\w y)$ is a monotone and submodular function on $\w x$.
		
		\item \label{appendix:clm:submodular-for-v}
		For any fixed  $\w s \in[0,\infty)^m$, $\f( \w s,\w y) $ is a monotone and submodular function on $\w y$.
		
		%

	\end{enumerate}
\end{lemma}

%
	\renewcommand{\A}{\mathcal{R}}
	\newcommand{\bvec}{\boldsymbol \beta}
%
%
%
	
\begin{proof}
	\paragraph{Property~(1)}

For any fixed set $Y$, denote the random subset of $Y$ containing 
	nodes activated by $X$ as $\A_Y(X)$ where the randomness comes from 
	probabilistic edges in the bipartite graph $B$ which will be 
	sampled according to the biadjacency matrix $M$.
Then $\sigma(\w x, \w y)= \sigma(X,Y) = \E{\U(\A_Y(X))}$. 
	Since $\E{\U(\A_Y(\cdot))}$ is a composition of monotone functions ($\E{\cdot}$, $\U(\cdot)$, and $\A_Y(\cdot)$), it is also monotone.
	To show that $\E{\U(\A_Y(\cdot))}$ is also submodular,
	we prove $\E{\U(\A_Y(W))}+\E{\U(\A_Y(T))} \geq  \E{\U(\A_Y(W\cup T))} + \E{\U(\A_Y(W\cap T))}$, for any $W,T \subseteq U$.
	
	Consider the fixed $Y$ and fix any realization of live-edge bipartite graph $B_L$ when all edges are sampled according to $M$.
	We use notation $\A_{B_L, Y}(X)$ to denote the set of all nodes in $Y$ 
	reachable from $X$ in graph $B_L$.
	Next, it is enough to prove for all possible live-edge graph $B_L$ and set $Y \subseteq V$, 
	\begin{equation}
	\U(\A_{B_L,Y}(W))+\U(\A_{B_L,Y}(T)) \geq \U(\A_{B_L,Y}(W\cup T)) + \U(\A_{B_L,Y}(W\cap T)). \label{appendix:eq:property-1}
	\end{equation} 
	For simplicify, we omit the subscript of notation, and let $\A = \A_{B_L,Y}$.
	
Since $\U(\cdot)$ is also a submodular function, we have
	\begin{equation}
	\U(\A(W))+\U(\A(T)) \geq  \U(\A(W)\cup \A(T)) + \U(\A(W)\cap \A(T)). \label{appendix:eq:property-1-1}
	\end{equation} 
Observe that
	\begin{align*}
	& \A(W\cup T) = \bigcup_{u\in W\cup T} \A(u) = \left(\bigcup_{u\in W} \A(u) \right) \cup \left(\bigcup_{u\in  T} \A(u) \right) =  \A(W)\cup \A(T); \\
	& \A(W\cap T) \subseteq \A(W), \A(T), \hbox{ and hence } \A(W\cap T) \subseteq \A(W)\cap \A(T). 
	\end{align*}
	Since $\U(\cdot)$ is  a monotone function, we know 
	$$
	\U\left(\A(W)\cup \A(T)\right)\ge \U(\A(W\cup T)), \hbox{ and } \U\left(\A(W)\cap \A(T)\right) \ge \U(\A(W\cap T)).
	$$
	Together with \eqref{appendix:eq:property-1-1}, Inequality~\eqref{appendix:eq:property-1} can be derived, which 
	completes the proof of Property~(1).

	\paragraph{Property~(2)}
	Note that, $\hat \sigma(\w{s},\w{y}) = \sum_{ Z\subseteq V} \Big[\rho(Z\cap Y) \cdot \Pr_{\w s, \w y}[Z] \Big]
$  where $\Pr_{\w s, \w y}[Z]$ denotes the probability that $Z$ is sampled out  from $V$  according to $\hF(\w s,\w y)$. Since for every fixed set $Z\subseteq V$,  $\rho(Z\cap Y)$ is a monotone and submodular function on $Y$,  $\hat \sigma(\w{s},\w{y})$ is a weighted average over such funcitons.  Therefore $\hat \sigma(\w{s},\w{y})$ is also a monotone and submodular function. 
\end{proof}	

\subsection{Supporting background probabilities for consumer nodes}
The above extended model also allows us to consider the following
	extension: we assume that every node $v$ in the social graph $G$ has an background (activation) probability $b_v$, that is,  the probability that  $v$ can be activated as one of the influence spread sources is $b_v$, 
	independent of whether $v$ is selected as a consumer seed.
As a result, the set $Z$ of initially activated consumer nodes in $G$ comes
	from two sources: a node $v$ is in $Z$ either because $v$ is selected
	as a consumer seed in $Y$ and $v$ is activated by some provider seed
	in $X$ through the bipartite graph $B$ with bi-adjacency matrix $M$, or
	$v$ is activated independently by a background probability $p_v$.
Then the final influence spread is $\rho(Z)$ once $Z$ is determined.

This extension covers the realistic cases where a consumer may pay attention
	to the advertiser's campaign anyway (either from content providers or
	any other unspecified sources) whether or not she is selected as 
	a seed, but if she is selected as a seed, she will pay more attention to
	the selected content providers and her probability of propagating the
	campaign is boosted.

For convenience, denote the vector of all background probabilities as $\w b$. 
We  use notation $\sigma'(X,Y)$ to represent the influence spread of $X$ and $Y$ in the combined network with background probabilities. 
Then we have $\sigma'(X,Y)=\sum_{Z_0\subseteq V}\sum_{Z\subseteq Y} \Pr_{\w b}[Z_0]\Pr_{X}[Z]\rho(Z\cup Z_0)$, where
$\Pr_{\w b}(Z_0)$ is the probability that $Z_0$  is sampled out from  $V$
	as the intially activated node set according to $\w b$, and
	$\Pr_{X}(Z)$ is the probability that $Z$ is sampled out from $Y$
	as the initially activated node set activated by
	the provider seed set $X$ according to the matrix $M$. 


We now show that this extension $\sigma'(X,Y)$ can be treated as a
	special case of the general model defined in Section~\ref{sec:generalmodel}.
By the definition of $\sigma'(X,Y)$, we have
$$\sigma'(X,Y)=\sum_{Z_0\subseteq V}\sum_{Z\subseteq Y} 
\Pr_{\w b}[Z_0]\Pr_{X}[Z]\rho(Z\cup Z_0)
=\sum_{Z\subseteq Y} \Pr_{X}[Z] \sum_{Z_0\subseteq V}
	\Pr_{\w b}[Z_0]\rho(Z\cup Z_0).$$
Define $\rho'(Z) = \sum_{Z_0\subseteq V} \Pr_{\w b}[Z_0]\rho(Z\cup Z_0)$,
	then we have $\sigma'(X,Y)= \sum_{Z\subseteq Y} \Pr_{X}[Z] \rho'(Z) $.
Hence, $\sigma'(X,Y)$ can be viewed as the final influence spread in the
	general model defined in Section~\ref{sec:generalmodel} with
	$\rho'$ as the influence spread in the social network $G$.
Since $\rho$ is monotone and submodular, it is straightforward to check that 
	$\rho(Z\cup Z_0)$ is monotone and submodular in $Z$ for any $Z_0$, and thus
	$\rho'(Z)$ as a non-negative linear combination of
	$\rho(Z\cup Z_0)$'s is also monotone and submodular.
Therefore, the extension with background probabilities can indeed be treated
	as a special case of the general model, and thus the algorithm and result
	in Section~\ref{appendix:sec:alg} can cover this further extension.
The only thing is that to compute $\rho'(Z)$, we may need to combine Monte Carlo
	simulations for set $Z_0$ 
	with the computation oracle for $\rho(\cdot)$ to get an
	accurate estimate for $\rho'(Z)$.

%% file: appendix-APM.tex
{\tiny }
\section{Hardness of approximation 
	result for Acceptance Probability Maximization}
\label{sec:APM}

In this appendix, we apply ideas from our hardness for AIM to prove the hardness of
	approximation result for the problem of 
	{\em acceptance probability maximization (APM)} studied 
	by \cite{YHLC13} in the context of active friending.
In APM, an initiator $s$ tries to
	find $k$ nodes in a social network to send friending requests to in
	order to maximize the eventual acceptance probability of a target node $t$,
	when $s$ finally sends a friending request to $t$.
In this model, if $s$ sends a friending request to a non-friend
	$v$ in the network, then the common friends of $s$ and $v$ would each
	independently influence $v$ to accept the request from $s$; once
	$v$ accepts the request, the influence can further propagate to 
	$v$'s friends who also receive friending requests from $s$.
Technically, the diffusion is formulated as following the
	independent cascade (IC) model and the maximization problem is equivalent
	to finding a subgraph such that the activation probability of target $t$
	is maximized when diffusion only propagates in this subgraph from seed
	nodes to $t$, where seed nodes are essentially the original friends of
	source node $s$.
We formally restate the APM problem below.

\begin{definition}[Acceptance Probability Maximization (APM) \cite{YHLC13}]\label{prob:original}
Given a graph $G=\left(V,E\right)$ with independent probabilities $p_{e}$
	on the edges, seed set $S\subseteq V$, a target node $t\in V \setminus S$,
	and a budget $B$. 
The output of APM is a subset $W\subseteq V$ of size 
	$\left|W\right|\leq B$.
Let $G(S\cup W \cup \{t\})$ be the subgraph of $G$ induced by
	nodes in $S\cup W \cup \{t\}$, and suppose that
	influence diffusion in $G(S\cup W \cup \{t\})$ follows the independent cascade model
	with edge probabilities $p_e$ for every edge $e$ in the subgraph 
	$G(S\cup W \cup \{t\})$.
The goal of APM is to maximize the activation probability of $t$ 
	when influence diffusion is from the seed set $S$ and is restricted to the subgraph $G(S\cup W \cup \{t\})$.

\end{definition}

The APM problem bears similarities to the AIM problem --- both are
	maximizing the effect of influence diffusion, both need
	to select certain number of nodes with respect to the budget constraint,
	and the influence diffusion in both problems are restricted in some way
	by the selected nodes.
However, they differ in two important aspects: 
	first, APM restricts the influence diffusion within the selected subgraph,
	while AIM only restricts diffusion from the selected seed providers
	to selected seed consumers, but from seed consumers, the diffusion can reach
	all other nodes in the social network;
	second, APM uses one budget for selecting the subgraph, while AIM
	uses two separate budgets on seed providers and seed consumers respectively.

The differences in the two problems prevent us from providing a black box
	reduction between the two problems, but their similarities allow
	us to apply the techniques from AIM hardness to APM hardness.
Moreover, by exploiting the fact that APM restricts the diffusion to 
	the selected subgraph from an arbitrary input graph, we are able to
	amplify the constant-factor hardness result of AIM 
	(Theorem \ref{thm:np-hardness}) to get an even
	stronger inapproximability result for APM:

%

\begin{theorem} \label{thm:APMhardness}
For any constant $\varepsilon > 0$,  APM over general graph $G$ is \NP-hard to 
	approximate to within factor $2^{-n^{(1-\varepsilon)}}$, where $n$ is the number of nodes in $G$. 
\end{theorem}

The rest of this appendix is devoted to the proof of Theorem \ref{thm:APMhardness}.
In the next subsection we prove that in the special case of a three-layer graph
	(when disregarding the single source node $s$ and the single target
	node $t$), APM is \NP-hard to approximate to within any constant factor (Lemma \ref{lem:3-APM}). This proof is almost identical to the proof of our main hardness result for AIM (Theorem \ref{thm:np-hardness}).
Then, in Subsection \ref{exp-hardness} we concatenate $n^{1-\varepsilon}$ instances of three-layer APM to achieve exponential hardness.

\subsection{Constant factor hardness for three-layer APM}

In this subsection we prove that in the special case of a three-layer graph
	(when excluding the single seed node $s$ and the single
	target node $t$), APM is \NP-hard to approximate to within any constant factor. In fact, it will be convenient to prove the following slightly stronger bi-criteria inapproximability:

\begin{lemma}\label{lem:3-APM}
Let $\alpha > 0$ be any constant. Given a budget $B$ and a three-layer graph $G=(V,E)$ (s.t. $V = \{s\} \cup U \cup V_1 \cup V_2 \cup \{t\}$), it is \NP-hard to distinguish between the following:
\begin{description}
\item[Completeness] the associated APM instance has value at least $1/3$; and
\item[Soundness] even with budget $2B$, the associated APM instance has value at most $\alpha$.
\end{description}
\end{lemma}

\begin{proof}
Our proof is very similar to the proof of Theorem \ref{thm:np-hardness}.
The main difference is that for the soundness we need to rule out solutions that perform much better using additional budget.
This additional budget comes from having $2B$ budget instead of $B$; from allowing additional budget to seed the nodes in $V_2$ (in AIM those nodes are ``free''); and from transferring budgets between layers (in AIM the partitioning of budget between layers is fixed by the instance). 
In particular, to overcome the latter problem we create many copies of $U$ and set the parameters so that the optimal solution uses approximately the same fraction of the budget in each layer.
The result will follow by observing that increasing the budget on any layer by a constant factor increases the probability of acceptance by at most a constant value. 
While the proof in this section is self-contained, 
we encourage the reader to refer back to the description of Feige's $k$-prover proof system in Section \ref{sec:Three-layers}; 
in particular, $k$, $l$, $Q$, and $R$ below are parameters of the $k$-prover proof system.

\paragraph{Construction}

We let the seed set contain a single vertex $S=\{s\}$ (this is without loss of generality).
We then construct three layers: $U,V_1,V_2$.
The source node $s$ is connected to all nodes $u \in U$ with probability $1$,
and each node in $V_{2}$ is connected to the target node $t$ with probability $1/\big((1-1/\e) R\big)$.

Going back to the $k$-prover system, 
the top layer $U$ corresponds to triplets of provers' answers
to questions; the middle layer $V_{1}$ corresponds to assignments
to variables -{\em distinguished and non-distinguished}- that may
appear in the verifier's question to any of the provers; finally,
the bottom layer corresponds to the random strings of the verifier.
All the edges go from the top to the middle layer, or from the middle to the bottom layer. 

More specifically, for each triplet $\left(q,a,i\right)$ of (question,
answer, prover) we have $\eta = R / (kQ)$ corresponding nodes in $U$. 
For each pair $\left(r,\overline{a_{r}}\right)$
of (verifier's random string, assignment to all $3l$ variables) we
have a node in $V_{1}$. 
Notice that this is different from \cite{feige98},
where the elements to be covered correspond to $\left(r,a_{r},i\right)$
with $a_{r}$ being the assignment only for the distinguished variables.
For every $h \in \left[\eta\right]$, the node $\left(q,a,i,h\right)$ is connected
to all the nodes $\left(r,\overline{a_{r}}\right)$ such that: $\left(q,i\right)\in r$,
and when restricting $\overline{a_{r}}$ to the variables specified
by $\left(q,i\right)$, it is equal to $a$. In particular, for each
$i$, each $\left(r,\overline{a_{r}}\right)$ corresponds to only
one $\left(q,a,i\right)$ (and thus $\eta$ different nodes $\left(q,a,i,h\right)$). 
Finally, all the edges from $U$ to $V_{1}$ 
have probability $1/(\eta k)$.

For each random string $r$, 
we have one node in the bottom layer, $V_{2}$.
The node corresponding to each $r$ is connected to all the nodes 
$\left(r,\overline{a_{r}}\right)$ in $V_{1}$ with probability $1$. 
The role of this layer is to force any good assignment to spread 
its budget across the different random strings
(i.e. make sure that the provers answer all the questions).

Finally, we set the budget $B = 3R$.
See Table \ref{tab:Summary-of-notation2} for a summary of notation.

\begin{table}
\caption{Summary of notation for Lemma }\label{tab:Summary-of-notation2}
\vspace{0.7cm}
\begin{tabular}{|c|c|c|}
\hline 
$\left(q,a,i,h\right)$ & question, answer, prover, copy & $1$ vertex in $U$\tabularnewline
\hline 
$\left(q,a,i\right)$ & question, answer, prover & $\eta$ vertices in $U$\tabularnewline
\hline 
$\left(q,i,h\right)$ & question, prover, copy & $2^{3l/2}$ vertices in $U$\tabularnewline
\hline 
$\left(q,i\right)$ & question, prover & $2^{3l/2}\cdot\eta$ vertices in $U$\tabularnewline
\hline 
$\left(r,h\right)$ & random string, copy & %
\begin{tabular}{c}
$k\cdot2^{3l/2}$ vertices in $U$\tabularnewline
($\forall$ $\left(q,i\right)\in r$ and $a\in\left\{ 0,1\right\} ^{3l/2}$)\tabularnewline
\end{tabular}\tabularnewline
\hline 
$\left(r,\overline{a_{r}}\right)$ & random string, assignment to all $3l$ variables & $1$ vertex in $V_{1}$\tabularnewline
\hline 
$r$ & random string & $1$ vertex in $V_{2}$\tabularnewline
\hline 
\end{tabular}
\end{table}

\subsubsection{Completeness}

Given a satisfiable assignment to the 3SAT-5 formula, 
in the top layer we let $W \cap U$ be the
$\eta kQ$ nodes that correspond to the same assignment.
Because they all correspond to the same assignment, 
for each random string $r$, 
all $\eta k$ corresponding nodes in $S$ are connected
to the common node $\left(r,\overline{a_{r}}^{*}\right)$. 
In the middle layer, we let $W \cap V_1$ be the set of these $R$ nodes 
(i.e. $\left(r,\overline{a_{r}}^{*}\right)$ for $r \in R$).
Before sampling the edges, each $\left(r,\overline{a_{r}}^{*}\right)$
has $\eta k$ neighbors in $W \cap U$. After sampling, the probability
that there is a path from $W \cap U$ to $\left(r,\overline{a_{r}}^{*}\right)$
is $1-\left(1-\frac{1}{k}\right)^{k}\approx1-1/\e$.
In particular, with high probability approximately $(1-1/\e) R$ of the nodes in $W \cap V_1$ are activated (e.g. via Chernoff bound).

Finally, we let $W \cap V_2  = V_2$. With high probability, approximately $(1-1/\e) R$ of them are activated.
Thus the probability that $t$ is activated is given by 
$1-\left(1-1/\big((1-1/\e) R\big)\right)^{(1-1/\e) R}\approx1-1/\e$.

\subsubsection{Soundness}
Let $OPT$ denote the optimum value (using budget $2B=6R$ on a ``no'' instance),
and let $OPT(B_1,B_2,B_3)$ denote the optimum value among assignments that spend budget $B_i$ on the $i$-th layer.
Clearly, $OPT \leq OPT(6R,6R,6R)$ since adding nodes can only increase the value.
In fact, any solution can spend at most $R$ budget on the last layer, so $OPT \leq OPT(6R,6R,R)$.
Now, observe that if we fix $W \cap (V_1 \cup V_2)$, the probability of activating $t$ is a monotone submodular function of $W \cap U$.
Thus $OPT \leq \frac{1}{6}OPT(R,6R,R)$. Similarly, when we fix $W \cap (U \cup V_2)$, the probability of activating $t$ is a monotone submodular function of $W \cap V_1$.
Therefore, $OPT \leq \frac{1}{36}OPT(R,R,R)$. 
In particular, it suffices to show that $OPT(R,R,R)$ is bounded by an arbitrarily small constant.

In an unsatisfiable instance, any two provers agree for at most a
$\left(2^{-cl}\right)$-fraction of the random strings. 
We will show
in Lemma \ref{lem:good-random-strings2} that there are 
at most $\left(2\cdot2^{-\left(1/3\right)cl} \cdot R\right)$ 
{\em good} random strings $r$, which are strings $r$ such that there is a node $\left(r,\overline{a_{r}}\right)$
with more than $2\eta$ neighbors in $W \cap U$. 
Since for each random
string $r$ there is only one node in $V_{2}$, each of the good random
strings contributes at most one to the number of activated neighbors of $t$. 
Before sampling the edges, 
any node that does not correspond to a good random string
has at most $2 \eta$ neighbors in $W \cap U$. 
After sampling the edges between $U$ and $V_1$, the probability that any such node has a neighbor in $W \cap U$ is at most $2/k$. 
Again, each such node can contribute at most one to the number of activated neighbors of $t$. 
In total, the number of activated neighbors of $t$ is bounded by:
\[ 
\big(\text{\# of good strings}\big) + \frac{2}{k}\big(\text{\# of bad strings}\big) \leq 2\cdot2^{-\left(1/3\right)cl} \cdot R + \frac{2}{k}  R < \frac{3}{k}  R.
\]

Recall that each neighbor activates $t$ with probability $1/\big((1-1/\e) R\big)$. 
Therefore, by union bound, the probability that any of the $\frac{3}{k}  R$ activated neighbors propagates to $t$ is at most $\frac{3}{k(1-1/\e)} < 5/k$.

\end{proof}

\begin{lem}
\label{lem:good-random-strings2}There are at most $\left(2\cdot2^{-\left(1/6\right)cl}\cdot k^{2}\cdot R\right)$
good random strings.\end{lem}
\begin{proof}
Intuitively, any $\left(r,\overline{a_{r}}\right)$ which has $2\eta$
neighbors in $W \cap U$ corresponds to an agreement of at least two
provers - and therefore should be a rare event. In order to turn this
intuition into a proof, we must rule out solutions that distribute
the budget in an uneven manner that does not correspond to answers
of provers to verifier's questions. Fix any assignment to the ``no'' instance. 
In the next few paragraphs, we repeatedly
apply Markov's inequality to bound the number of: ``heavy $\left(q,i\right)$''
for which the assignment allocates $2^{\left(1/6\right)cl}$-times more
than the expected budget; ``heavy $\left(q,i,h\right)$'',
for which $2^{\left(1/3\right)cl}$-times
more than the expected budget is allocated; and ``good $\left(r,h\right)$''
for which two provers agree, i.e. some node $\left(r,\overline{a_{r}}\right)$
has more than one neighbor $\left(q,a,i,h\right)$ in $W \cap U$.

For any prover $i$, there are at most $\eta kQ$ corresponding nodes
in $W \cap U$, so at most $\eta k$ in expectation over $q$. By Markov's
inequality, for at most a $2^{-\left(1/6\right)cl}$-fraction of $q$'s,
more than $2^{\left(1/6\right)cl}\cdot\eta k$ nodes belong to $W \cap U$;
we call those $\left(q,i\right)$'s {\em heavy}, and {\em light}
otherwise. We henceforth focus on bounding the number of good random
strings that correspond only to light $\left(q,i\right)$'s.
\[
\Pr_{r}\left[\exists i:\mbox{ \ensuremath{\left(q,i\right)}\,\ is heavy}\right]\leq2^{-\left(1/6\right)cl}\cdot k\mbox{.}
\]

Recall that for each triplet $\left(q,a,i\right)$,
we have $\eta$ nodes in $U$ (with identical neighborhoods).
For $1\leq h\leq\eta$,
we label the $h$-th such node by $\left(q,a,i,h\right)$. 
Fix any light $\left(q,i\right)$. For each $h$, in expectation, $W \cap U$
contains at most $2^{\left(1/6\right)cl}\cdot k$ nodes $\left(q,a,i,h\right)$.
Using Markov's inequality again, for at most a $2^{-\left(1/6\right)cl}$-fraction
of the $h$'s, $W \cap U$ contains more than $2^{\left(1/3\right)cl}\cdot k$
nodes $\left(q,a,i,h\right)$. We abuse notation and call any such
triplet $\left(q,i,h\right)$ {\em heavy}, and {\em light} otherwise.
In particular, for any $r$ such that all the corresponding $\left(q,i\right)$'s
are light, at most a $2^{-\left(1/6\right)cl}$-fraction of the corresponding
$\left(q,i,h\right)$'s are heavy. For each heavy $\left(q,i,h\right)$,
any $\left(r,\overline{a_{r}}\right)$ has only one neighbor $\left(q,a,i,h\right)$.
Thus to each $\left(r,\overline{a_{r}}\right)$, all the heavy $\left(q,i,h\right)$'s
together contribute at most $2^{-\left(1/6\right)cl}k \cdot \eta$
neighbors in $W \cap U$. We henceforth ignore the heavy $\left(q,i,h\right)$'s,
and add these $2^{-\left(1/6\right)cl}k \cdot \eta$ nodes at the end.
\[
\forall \left(r,\overline{a_{r}}\right) \;\;\; \#\Big\{\left(q,a,i,h\right) \left(q,a,i,h\right)\in\mathcal{N}\left(r,\overline{a_{r}}\right)\cap (W \cap U)\mbox{ and $\left(q,i,h\right)$ is heavy}\Big\}\leq2^{-\left(1/6\right)cl}k \cdot \eta.
\]

Consider only light $\left(q,i,h\right)$'s. Then for each $h$ and
light $\left(q,i\right)$, there are at most $2^{\left(1/3\right)cl}\cdot k$
nodes $\left(q,a,i,h\right)$ in $W \cap U$. In other words, for each
$h$, each prover has at most $2^{\left(1/3\right)cl}\cdot k$ answers
to each question. Since we started from an unsatisfiable instance,
we have that for any pair of provers, at most a $2^{-cl}\cdot\left(2^{\left(1/3\right)cl}\cdot k\right)^{2}$-fraction
of random strings have at least one pair of agreeing answers (Theorem \ref{thm:feige}). Keeping
$h$ fixed and summing over all pairs of provers, this corresponds
to a $\left(2^{-\left(1/3\right)cl}\cdot k^{4}\right)$-fraction of
random strings $r$ such that any node $\left(r,\overline{a_{r}}\right)$
has more than one neighbor $\left(q,a,i,h\right)$ in $W \cap U$. We
say that a pair $\left(r,h\right)$ is {\em good} if for some $\overline{a_{r}}$,
the node $\left(r,\overline{a_{r}}\right)$ has more than one neighbor
$\left(q,a,i,h\right)$ in $W \cap U$. 
\[
\Pr_{r,h}\left[\mbox{\ensuremath{\left(r,h\right)}\,\ is good}\right]\leq2^{-\left(1/3\right)cl}\cdot k^{4}\mbox{.}
\]

Finally, for each random string $r$, in expectation, at most a $2^{-\left(1/3\right)cl}\cdot k^{4}$-fraction
of the $h$'s satisfy $\left(r,h\right)$ is good. Applying Markov's
inequality one more time, we have that for at most a $2^{-\left(1/6\right)cl}$-fraction
of the $r$'s, for more than a $2^{-\left(1/6\right)cl}\cdot k^{4}$-fraction
of the $h$'s, $\left(r,h\right)$ is good. We claim that these $r$'s,
together with the ones that correspond to heavy $\left(q,i\right)$'s,
are the only good random strings. Notice that there are at most $2\cdot2^{-\left(1/6\right)cl} k \cdot R$
of them.

It is left to prove that if $\left(r,h\right)$ is good for at most
a $2^{-\left(1/6\right)cl}\cdot\eta k^{4}$ of the $h$'s, then $r$
cannot be a good random string. For each $\left(r,\overline{a_{r}}\right)$,
each good $\left(r,h\right)$ contributes at most $k$ neighbors in
$W \cap U$. Together with additional $2^{-\left(1/6\right)cl}\eta k$
neighbors due to heavy $\left(q,i,h\right)$'s and a single neighbor
for each other $h$, we have that the number of neighbors of $\left(r,\overline{a_{r}}\right)$
in $W \cap U$ is at most 
\[
\left(1+2^{-\left(1/6\right)cl}\cdot k^{5}+2^{-\left(1/6\right)cl}k\right)\eta<2\eta.
\]
\end{proof}

\subsection{Exponential factor hardness}\label{exp-hardness}

We are now ready to complete the proof of Theorem \ref{thm:APMhardness}.
We concatenate $n^{1-\varepsilon}$ copies of the hard 3-layer APM instance, each of size $n^{\varepsilon}$.
(So that the total number of nodes is $n^{1-\varepsilon} \cdot n^{\varepsilon} = n$, 
and the blowup in size is polynomial in $n^{\varepsilon}$, for any constant $\epsilon$.)
Specifically, by concatenation we mean that we identify $t_i$, the target node of the $i$-th copy,
with $s_{i+1}$, the source node of the $i+1$-th copy.
The total budget is set to $n^{1-\varepsilon}(3R+1)$.

\begin{description}
\item[Completeness] If we can achieve value $1/2$ on each copy, the final activation probability of $t_{n^{1-\varepsilon}}$ is $2^{-n^{1-\varepsilon}}$.
\item[Soundness] We can allocate budget greater than $6R$ to at most half the instances. 
On the other half of the instances we would achieve value at most $\alpha$, 
where $\alpha$ is an arbitrarily small constant which depends on our instantiation of the 3-layer APM (in particular, $\alpha = 1/16$ suffices).
Therefore, the final activation probability is at most $\alpha^{n^{1-\varepsilon}/2} \leq 2^{-n^{1-\varepsilon}} \cdot 2^{-n^{1-\varepsilon}}$.
\end{description}
\qed

\paragraph{Remark}
One can easily generalize the APM problem to support a target set $T$
	of nodes with the goal of maximizing the expected number of active nodes in the intersection of the target set and the selected set, defined as APM-m problem
	below.
\begin{definition}\label{prob:general}[APM-m]
Given a graph $G=\left(V,E\right)$ with independent probabilities $p_{e}$
	on the edges, seed set $S\subseteq V$, target set 
	$T \subseteq V \setminus S$, and a budget $B$. 
The problem of APM-m is to find a subset $W\subseteq V$ of size 
	$\left|W\right|\leq B$.
Let $G(S\cup W)$ be the subgraph of $G$ induced by
	nodes in $S\cup W$, and suppose that
	influence diffusion in $G(S\cup W)$ follows the independent cascade model
	with edge probabilities $p_e$ for every edge $e$ in the subgraph $G(S\cup W)$.
The goal of APM-m is to maximize the expected number of active nodes
	in $T\cap W$ 
	when influence diffusion is from the seed set $S$ 
	and is restricted to the subgraph $G(S\cup W)$.
\end{definition}
Since APM-m is a generalization of APM with a single target, the
	near-exponential hardness of APM directly applies to this generalization.
We further remark that the proof of 
	the constant factor hardness of APM for three-layer graphs can be
	adapted to show that the constant factor hardness of APM-m
	for three-layer graphs (with one additional node as the single seed
	connecting to all first layer nodes with edge probability $1$, and the third-layer nodes as the targets).

%

%